\documentclass[11pt,reqno]{article}
\usepackage{amssymb,amscd,amsmath,amsthm}
\usepackage{graphicx}

\def\natural{\mathbf{N}}

\newtheorem{theorem}{Theorem}[section]
\newtheorem{proposition}[theorem]{Proposition}

\newtheorem{lemma}[theorem]{Lemma}
\newtheorem{definition}[theorem]{Definition}

\newtheorem{remark}[theorem]{Remark}

\def\cb{{\mathcal B}}
\def\�{{\mathcal C}}

\def\ce{{\mathcal E}}

\def\cw{{\mathcal W}}

\def\bc{{\mathbb C}}
\def\bh{{\mathbb H}}

\def\bn{{\mathbb N}}

\def\br{{\mathbb R}}

\def\frak{\mathfrak}

\def\ga{{\frak A}}

\def\a{\alpha}
\def\b{\beta}

\def\tr{{\rm Tr}}
\def\L{\Lambda}
\def\G{\Gamma}
\def\ce{\mathcal E}

\def\ffi{\varphi}

\def\Tr{\mathrm{Tr}}
\def\<{\langle}
\def\>{\rangle}

\def\1{\mathbf{1}}

\def\ve{\varepsilon}
\def\cw{\cal W}

\def\cal{\mathcal}

\def\s{\sigma}
\def\bh{\mathbf{h}}
\def\bs{\mathbf{s}}

\def\id{{\bf 1}\!\!{\rm I}}

\begin{document}

\begin{center}
{\Large {\bf Phase Transitions for quantum Ising model with competing $XY$-interactions on  a Cayley tree}}\\[1cm]
\end{center}

\begin{center}
{\large {\sc Farrukh Mukhamedov}}\\[2mm]
\textit{ Department of Mathematical Sciences,\\
College of Science, United Arab Emirates University,  \\
P.O. Box 15551,Al Ain, Abu Dhabi, UAE}\\
E-mail: {\tt far75m@yandex.ru, \ farrukh.m@uaeu.ac.ae}
\end{center}

\begin{center}

{\sc Abdessatar Barhoumi}\\
\textit{Department of Mathematics\\
Nabeul Preparatory Engineering Institute\\
Campus Universitairy - Mrezgua - 8000 Nabeul,\\
Carthage University, Tunisia}\\
E-mail: {\tt abdessatar.barhoumi@ipein.rnu.tn}\\
\end{center}

\begin{center}
{\sc Abdessatar Souissi}\\
\textit{
College of Business Administration,\\
Qassim university, Buraydah, Saudi Arabia}\\
E-mail: {\tt a.souaissi@qu.edu.sa}\\
\end{center}

\begin{center}
{\sc Soueidy EL Gheteb }\\
\textit{
Department of Mathematics,\\
Carthage University, Tunisia}\\
E-mail: {\tt elkotobmedsalem@gmail.com}\\
\end{center}


\begin{abstract} The main aim of the present paper is to establish the existence of a
phase transition for the quantum Ising model with competing $XY$ interactions within the quantum Markov chain (QMC) scheme. In this scheme, we employ the $C^*$-algebraic approach to the phase transition problem.
Note that these kinde of models
do not have one-dimensional analogues, i.e. the considered model
persists only on trees. It turns out that if the Ising part interactions vanish then the model with only competing $XY$-interactions  on the Cayley tree of order two does not have a phase transition.
By phase transition we mean the existence of two distinct QMC
which are not quasi-equivalent and their supports do not overlap. Moreover,  it is also shown that the QMC associated with the model have clustering property which implies that the von Neumann algebras corresponding to the states are factors.
\end{abstract}

\section{Introduction }\label{intr}

The study of magnetic systems with competing interactions in ordering is a fascinating
problem of condensed matter physics. One of the
most canonical examples of such systems are frustrated Ising models which demonstrate a plethora of critical properties\cite{Chak,Lieb,MS}. The
frustrations can be either geometrical or brought about the next-nearest neighbor NNN interactions. Competing interactions frustrations can result new phases, change the Ising universality class, or even destroy the order at all.
Another interesting aspect of the criticality in the frustrated Ising models is an appearance of quantum critical points at spacial frustration points of model's high degeneracy, and related quantum phase transitions \cite{Sach}.  The Ising models with frustrations can be thought as perturbation of the classical Ising model.  If the perturbation terms do not commute with the Ising pieces, it outcomes quantum effects. In particular case, if the perturbation is the XY interaction, then the model become more interesting (see \cite{AAC,CG,MS,Ow} for a systematic study (physical approach)  of the Ising model with quantum
frustration on 2D lattices). However, a rigorous (mathematical) investigation of  the quantum Ising model with competing $XY$ interactions does not exist yet in the literature.  We notice that $XY$-interactions are truly quantum, ( i.e.
contain pieces not commuting with each other). 
 In the present paper, we propose to investigate the phase transition problem for the mentioned model on the Cayley tree or Bethe lattice \cite{Ost} within quantum Markov chains (QMC) scheme.  Here, the QMC scheme is based on the 
 $C^*$-algebraic approach.  We notice that the Ising model with Ising type competing interactions (with commuting interactions) has been recently studied in \cite{MBS161,MBS162,MR1,MR2} by means of QMC.  As we mentioned, in the current paper, the
 commuting interactions are non-commutative, and this makes big difference between those papers. 

On the other hand, our investigation will allow to construct quantum analogous of Markov fields (see \cite{D,Geor,[Pr],[Sp75],Spa}) which is one of the basic problems in quantum probability.  We notice that quantum Markov fields naturally
appear in quantum statistical mechanics and quantum fields theories \cite{DW,DM}\footnote{The quantum analogues of Markov chains were first constructed in
\cite{[Ac74f]}, where the notion of quantum Markov chain  (QMC) on
infinite tensor product algebras was introduced. Later on, in \cite{fannes2}, finitely correlated states were introduced and studied, which are related to each other. However, satisfactory constructions of such kind of fields were not established, since most of the fields were considered over the integer lattices \cite{[AcFi01a],AF03}.}.

We point out,  even in classical setting, for models over integer lattices, there do not exist
analytical solutions (for example, critical temperature) on such
lattices. Therefore, it was proposed \cite{Bax} to consider spin models on regular trees for which one can exactly calculate
various physical quantities. One of the simplest tree is a Cayley tree  \cite{Ost}.
In \cite{MBS161,MBS162} we have established that Gibbs measures of the Ising model with competing (Ising) interactions on a Cayley trees,  can be considered as  QMC. Note that if the perturbation vanishes then the model reduces to the classical Ising one which was also examined in  \cite{ArE} by means of $C^*$-algebra approach.

In the present paper, we are going to study the Ising model with $XY$-competing interactions on a Cayley trees of order two. We point out that this model has non-commutative interactions, i.e. $XY$ ones, therefore, the investigation of this model is tricky.
We notice that, in general, QMC do not have KMS property (see
\cite{AF03,BR}), therefore, general theory of KMS-states is not
applicable for such kind of chains. One of the main questions of
this paper is to know whether the considered model exhibits two
different QMC associated with the mentioned model on the Cayley trees.
Our main result is the
following one.

\begin{theorem}\label{Main}
For the Ising model with competing $XY$- interactions \eqref{1Kxy1},
\eqref{1Hxy1}, $J_{0}>0$, $J \in \mathbb{R}\setminus\{-J_0,J_0\}$, $\b>0$ on the Cayley tree of order two, the
following statements hold:
\begin{itemize}
\item if $\Delta(\theta)\leq 0$, then there is a unique QMC;
\item if $\Delta(\theta)>0$, then there occurs a phase transition,
\end{itemize}
where $$\Delta(\theta)=\frac{\theta^{2J_0}-\theta^{J_0}(\theta^{J}+\theta^{-J})-3}{\theta^{2J_0}-\theta^{J_0}(\theta^{J}+\theta^{-J})+1},~~
\theta=e^{2\beta}.$$
\end{theorem}

To establish result, we will prove Theorem \ref{5.3}, from which we conclude that there are three coexisting phases in the region  $\Delta(\theta)>0$, and one of it, i.e. $\varphi_{\alpha}$, survives in the region $\Delta(\theta)<0$. This leads us that
the state $\varphi_{\alpha}$ describes the disordered phase of the model, which shows a similar behavior with the classical Ising model \cite{B90,Roz2}.
In comparison with the Ising  model, we stress  that in the present model, we have a similar kind of phases (translation invariant ones) when  $\Delta(\theta)>0$. From Figure 2 (see below), one concludes that the phase transition
occurs except for a "triangular region". This shows how the competing interactions effect to the existence of other phases. Notice that if $J=0$, then we obtain the classical Ising model for which the existence of a disordered phase coexisting with
two ordered phases is well-known \cite{Bax, Roz2}.

On the other hand, we emphasize that
both problems, i.e. a construction and phase transitions are non-trivial and, to a
large extent, open. In fact, even if several definitions of
quantum Markov fields on trees (and more generally on graphs) have
been proposed, a really satisfactory, general theory is still
missing and physically interesting examples of such fields in
dimension $d\geq2$ are very few.

In order to get the existence of the phase transition (see \cite{MBS161}), one needs to check several conditions, and  one of them based on a notion of the
quasi-equivalence of quantum Markov chains which essentially uses $C^*$-algebraic approach and techniques. This situation totally differs from the classical (resp. quantum) cases,  where it is sufficient to prove the existence of at
least two different solutions (resp. KMS states) of associated renormalized equations
(see \cite{Roz2}). Therefore, even for classical models, to check the
existence of the phase transition (in the sense of our paper) is not a trivial problem. Here we mention  that the quasi-equivalence of product states (which correspond to the classical model without interactions) was a tricky job and investigated in  \cite{Mat0,PS}.  In this paper, we are considering more complicated states (which are QMC associated with the model) than product ones, and for these kind of states we are going to obtain their non-quasi equivalence. 
We will first show that these states have clustering property, and hence they are factor states. We point out that even this fact presents its own interests since these states associates with non-commutative
Hamiltonians having non-trivial interactions.

Let us outline the organization of the paper. After preliminary
information (see Section 2), in Section 3 we provide a general
construction of backward quantum Markov chains on Cayley tree. Moreover, in this section we give the definition of the phase
transition. Using the provided construction, in Section 4 we
consider the Ising model with competing $XY$-interactions on the Cayley
tree of order two. Section 5 is devoted to the existence of the
three translation-invariant QMC $\varphi_\a$, $\varphi_1$ and
$\varphi_1$ corresponding to the model. Section 6 is devoted to the
proof of Theorem \ref{Main}. In this section we will prove that
states $\varphi_1$ and $\varphi_2$ do not have overlapping
supports. Before, to establish their non-quasi equivalence, we first prove that these states have the clustering property. Section 7 we study a particular case $J_0=0$,  which means that we only have $XY$ interactions. In the considered setting, it turns out  that the phase transition does not occur.

\section{Preliminaries}

Let $\Gamma^k_+ = (L,E)$ be a semi-infinite Cayley tree of order
$k\geq 1$ with the root $x^0$ (i.e. each vertex of $\Gamma^k_+$
has exactly $k+1$ edges, except for the root $x^0$, which has $k$
edges). Here $L$ is the set of vertices and $E$ is the set of
edges. The vertices $x$ and $y$ are called {\it nearest neighbors}
and they are denoted by $l=<x,y>$ if there exists an edge
connecting them. A collection of the pairs
$<x,x_1>,\dots,<x_{d-1},y>$ is called a {\it path} from the point
$x$ to the point $y$. The distance $d(x,y), x,y\in V$, on the
Cayley tree, is the length of the shortest path from $x$ to $y$.

Recall a coordinate structure in $\G^k_+$:  every vertex $x$
(except for $x^0$) of $\G^k_+$ has coordinates $(i_1,\dots,i_n)$,
here $i_m\in\{1,\dots,k\}$, $1\leq m\leq n$ and for the vertex
$x^0$ we put $(0)$.  Namely, the symbol $(0)$ constitutes level 0,
and the sites $(i_1,\dots,i_n)$ form level $n$ (i.e. $d(x^0,x)=n$)
of the lattice (see Fig. 1).

\begin{figure}
\begin{center}
\includegraphics[width=10.07cm]{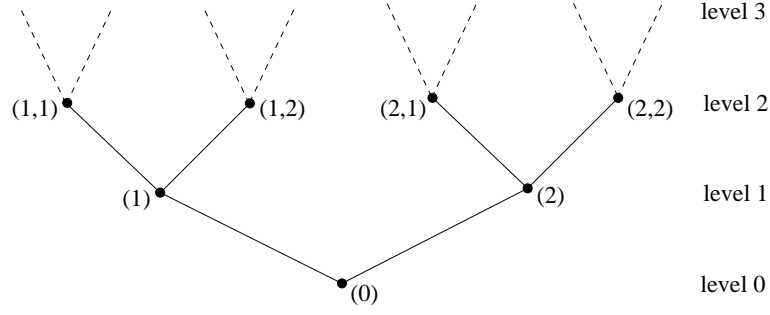}
\end{center}
\caption{The first levels of $\G_+^2$} \label{fig1}
\end{figure}

Let us set
\[
W_n = \{ x\in L \, : \, d(x,x_0) = n\} , \qquad \Lambda_n =
\bigcup_{k=0}^n W_k, \qquad  \L_{[n,m]}=\bigcup_{k=n}^mW_k, \
(n<m)
\]
\[
E_n = \big\{ <x,y> \in E \, : \, x,y \in \Lambda_n\big\}, \qquad
\Lambda_n^c = \bigcup_{k=n}^\infty W_k
\]
For $x\in \G^k_+$, $x=(i_1,\dots,i_n)$ denote
$$ S(x)=\{(x,i):\ 1\leq
i\leq k\}.
$$
Here $(x,i)$ means that $(i_1,\dots,i_n,i)$. This set is called a
set of {\it direct successors} of $x$.

Two vertices $x,y\in V$ is called {\it one level
next-nearest-neighbor  vertices} if there is a vertex $z\in V$
such that  $x,y\in S(z)$, and they are denoted by $>x,y<$. In this
case the vertices $x,z,y$ was called {\it ternary} and denoted by
$<x,z,y>$.

Let us define on $\G^k_+$ a binary operation
$\circ:\G^k_+\times\G^k_+\to\G^k_+$ as follows: for any two
elements $x=(i_1,\dots,i_n)$ and $y=(j_1,\dots,j_m)$ put
\begin{equation}\label{binar1}
x\circ
y=(i_1,\dots,i_n)\circ(j_1,\dots,j_m)=(i_1,\dots,i_n,j_1,\dots,j_m)
\end{equation}
and
\begin{equation}\label{binar2}
x\circ x^0=x^0\circ x= (i_1,\dots,i_n)\circ(0)=(i_1,\dots,i_n).
\end{equation}

By means of the defined operation $\G^k_+$ becomes a
noncommutative semigroup with a unit. Using this semigroup
structure one defines translations $\tau_g:\G^k_+\to \G^k_+$,
$g\in \G^k_+$ by
\begin{equation}\label{trans1}
\tau_g(x)=g\circ x.
\end{equation}
It is clear that $\tau_{(0)}=id$.

 The algebra of observables $\cb_x$ for any single site
$x\in L$ will be taken as the algebra $M_d$ of the complex
$d\times d$ matrices. The algebra of observables localized in the
finite volume $\L\subset L$ is then given by
$\cb_\L=\bigotimes\limits_{x\in\L}\cb_x$. As usual if
$\L^1\subset\L^2\subset L$, then $\cb_{\L^1}$ is identified as a
subalgebra of $\cb_{\L^2}$ by tensoring with unit matrices on the
sites $x\in\L^2\setminus\L^1$. Note that, in the sequel, by
$\cb_{\L,+}$ we denote the positive part of $\cb_\L$. The full
algebra $\cb_L$ of the tree is obtained in the usual manner by an
inductive limit
$$
\cb_L=\overline{\bigcup\limits_{\L_n}\cb_{\L_n}}.
$$

In what follows, by ${\cal S}({\cal B}_\L)$ we will denote the set
of all states defined on the algebra ${\cal B}_\L$.

Consider a triplet ${\cal C} \subset {\cal B} \subset {\cal A}$ of
unital $C^*$-algebras. Recall \cite{ACe} that a {\it
quasi-conditional expectation} with respect to the given triplet
is a completely positive (CP) linear map $\ce \,:\, {\cal A} \to
{\cal B}$ such that $ \ce(ca) = c \ce(a)$, for all $a\in {\cal
A},\, c \in {\cal C}$.

\begin{definition}[\cite{AOM}]\label{QMCdef}
A state $\varphi$ on ${\cal B}_L$ is called a {\it backward quantum
Markov chain (QMC)}, associated to $\{\L_n\}$, if for each
$\Lambda_n$, there exist a quasi-conditional expectation
$\ce_{\Lambda_n}$ with respect to the triplet
\begin{equation}\label{trplt1}
{\cal B}_{{\Lambda}_{n-1}}\subseteq {\cal
B}_{\Lambda_n}\subseteq{\cal B}_{\Lambda_{n+1}}
\end{equation}
and an initial state $\rho \in S(B_{\Lambda_0})$ such that:
\begin{equation}\label{dfgqmf}
\varphi = \lim_{n\to\infty} \rho_0 \circ
\ce_{\Lambda_0}\circ \ce_{\Lambda_{1}} \circ \cdots \circ
\ce_{\Lambda_n}
\end{equation}
in the weak-* topology.
\end{definition}

\begin{remark} We notice that in \cite{[AcFiMu07]} a more general definition of backward QMC is given on arbitrary quasi-local algebras.  \end{remark}

\section{Construction of Quantum Markov Chains}\label{dfcayley}

In this section we are going to provide a construction of a
backward quantum Markov chain which contains competing interactions.

Let us rewrite the elements of $W_n$ in the following lexicographic order (w.r.t. the coordinate system), i.e.
\begin{eqnarray*}
\overrightarrow{W_n}:=\left(x^{(1)}_{W_n},x^{(2)}_{W_n},\cdots,x^{(|W_n|)}_{W_n}\right).
\end{eqnarray*}

Note that $|W_n|=k^n$. In this lexicographic order, vertices 
$x^{(1)}_{W_n},x^{(2)}_{W_n},\cdots,x^{(|W_n|)}_{W_n}$ of $W_n$
are given as follows
\begin{eqnarray}\label{xw}
&&x^{(1)}_{W_n}=(1,1,\cdots,1,1), \quad x^{(2)}_{W_n}=(1,1,\cdots,1,2), \ \ \cdots \quad x^{(k)}_{W_n}=(1,1,\cdots,1,k,),\\
&&x^{(k+1)}_{W_n}=(1,1,\cdots,2,1), \quad
x^{(2)}_{W_n}=(1,1,\cdots,2,2), \ \ \cdots \quad
x^{(2k)}_{W_n}=(1,1,\cdots,2,k),\nonumber
\end{eqnarray}
\[\vdots\]
\begin{eqnarray*}
&&x^{(|W_n|-k+1)}_{W_n}=(k,k,,\cdots,k,1), \
x^{(|W_n|-k+2)}_{W_n}=(k,k,\cdots,k,2),\ \ \cdots
x^{|W_n|}_{W_n}=(k,k,\cdots,k,k).
\end{eqnarray*}

Analogously, for a given vertex $x,$ we shall use the following
notation for the set of direct successors of $x$:
\begin{eqnarray*}
\overrightarrow{S(x)}:=\left((x,1),(x,2),\cdots (x,k)\right),\quad
\overleftarrow{S(x)}:=\left((x,k),(x,k-1),\cdots (x,1)\right).
\end{eqnarray*}

In what follows, by $^\circ\prod$ we denote the lexicographic order,
i.e.
$$
^\circ\prod_{k=1}^n  a_k=a_1a_2\cdots a_n,
$$
where elements $\{a_k\}\subset {\mathcal{B}}_L$ are multiplied in
the indicated order. This means that we are not allowed to change
this order.

 Note that each vertex $x\in L$ has interacting vertices
$\{x, (x,1),\dots,(x,k)\}$. Assume that, to each edges $<x,(x,i)>$
($i=1,\dots,k$) an operators
$K_{<x,(x,i)>}\in\mathcal{B}_{x}\otimes\mathcal{B}_{(x, i)}$ is
assigned, respectively. Moreover, for each competing vertices
$>(x,i),(x,i+1)<$ and $<x,(x,i),(x,i+1)>$ ($i=1,\dots,k$) the
following operators are assigned:
 $$L_{>(x,i),(x,i+1)<}\in\mathcal{B}_{(x,i)}\otimes\mathcal{B}_{(x,i+1)}, \ \
 M_{(x,(x,i),(x,i+1))}\in\mathcal{B}_{x}\otimes\mathcal{B}_{(x,i)}\otimes\mathcal{B}_{(x,i+1)}.$$

We would like to define a state on $\mathcal{B}_{\Lambda_n}$ with
boundary conditions $\omega_{0}\in\mathcal{B}_{0,+}$ and
$\{h^{x}\in\mathcal{B}_{x,+}:\ x\in L\}$.

For each $n\in \natural$ denote
\begin{eqnarray}\label{K1}
&& A_{x,(x,1),\dots,(x,k)}=\big(^\circ\prod_{i=1}^k
K_{x,(x,i)}\big)\big(^\circ\prod_{i=1}^k
L_{>(x,i),(x,i+1)<}\big)\big(^\circ\prod_{i=1}^k
M_{(x,(x,i),(x,i+1))}\big),\\[2mm]
\label{K11} &&K_{[m,m+1]}:=\prod_{x\in \overrightarrow
W_{m}}A_{x,(x,1),\dots,(x,k)}, \ \ 1\leq m\leq n,\\[2mm]
 \label{K2} &&
\bh_{n}^{1/2}:=\prod_{x\in W_n}(h^{x})^{1/2}, \ \ \  \bh_n=\bh_{n}^{1/2}(\bh_{n}^{1/2})^*\\[2mm]
\label{K3}
&&{\mathbf{K}}_n:=\omega_{0}^{1/2} (^\circ\prod_{m=0}^{n-1}K_{[m,m+1]})\bh_{n}^{1/2}\\[2mm]
\label{K4} && \mathcal{W}_{n]}:={\mathbf{K}}_n^{*}{\mathbf{K}}_n
\end{eqnarray}
One can see that ${\cw}_{n]}$ is positive.

In what follows, by $\tr_{\L}:\cb_L\to\cb_{\L}$ we mean normalized
partial trace (i.e. $\tr_{\L}(\id_{L})=\id_{\L}$, here
$\id_{\L}=\bigotimes\limits_{y\in \L}\id$), for any
$\Lambda\subseteq_{\text{fin}}L$. For the sake of shortness we put
$\tr_{n]} := \tr_{\Lambda_n}$.

Let us define a positive functional $\ffi^{(n)}_{w_0,\bh}$ on
$\cb_{\Lambda_n}$ by
\begin{eqnarray}\label{ffi-ff}
\ffi^{(n)}_{w_0,\bh}(a)=\tr(\cw_{n+1]}(a\otimes\id_{W_{n+1}})),
\end{eqnarray}
for every $a\in \cb_{\Lambda_n}$. Note that here, $\tr$ is a
normalized trace on ${\cal B}_L$ (i.e. $\tr(\id_{L})=1$).

To get an infinite-volume state $\ffi$ on $\cb_L$  such that
$\ffi\lceil_{\cb_{\L_n}}=\ffi^{(n)}_{w_0,\bh}$, we need to impose
some constrains to the boundary conditions $\big\{w_0,\bh\big\}$
so that the functionals $\{\ffi^{(n)}_{w_0,\bh}\}$ satisfy the
compatibility condition, i.e.
\begin{eqnarray}\label{compatibility}
\ffi^{(n+1)}_{w_0,\bh}\lceil_{\cb_{\L_n}}=\ffi^{(n)}_{w_0,\bh}.
\end{eqnarray}



\begin{theorem}\label{compa} Assume that for every $x\in L$ and triple $\{x,(x,i),(x,i+1)\}$ $(i=1,\dots,k-1$) the operators
$K_{<x,(x,i)>}$, $L_{>(x,i),(x,i+1)<}$, $M_{(x,(x,i),(x,i+1))}$
are given as above. Let the boundary conditions $w_{0}\in {\cal
B}_{(0),+}$ and ${\bh}=\{h_x\in {\cal B}_{x,+}\}_{x\in L}$ satisfy
the following conditions:
\begin{eqnarray}\label{eq1}
&&\Tr(\omega_{0}h_{0})=1, \\
\label{eq2}
&&\Tr_{x]}\big({A_{x,(x,1),\dots,(x,k)}}^\circ\prod_{i=1}^k
h^{(x,i)}A_{x,(x,1),\dots,(x,k)}^{*}\big)=h^x, \ \ \textrm{for
every}\ \ x\in L,
\end{eqnarray}
where as before $A_{x,(x,1),\dots,(x,k)}$ is given by \eqref{K1}.
Then the functionals $\{\ffi^{(n)}_{w_0,\bh}\}$ satisfy the
compatibility condition \eqref{compatibility}. Moreover, there is
a unique backward quantum d-Markov chain $\ffi_{w_0,{\bh}}$ on
$\cb_L$ such that
$\ffi_{w_0,{\bh}}=w-\lim\limits_{n\to\infty}\ffi^{(n)}_{w_0,\bh}$.
\end{theorem}

\begin{proof}
Let us check that the states  $\ffi^{(n,b)}_{w_0,{\bh}}$ satisfy the compatibility condition. For $a \in B_{\Lambda_n}$,
we have:

$\ffi^{(n+1)}_{w_0,{\bh}}(a\otimes \id_{\cw_{n+1}})=\tr(\cw_{n+2]}(a\otimes\id_{\Lambda_{[n+1,n+2]}}))$\\[2mm]
$=\tr(K_{[n,n+1]}^{*}K_{[n-1,n]}^{*}\cdots K_{[0,1]}^{*} w_0 K_{[0,1]} K_{[1,2]}
\cdots K_{[n+1,n+2]}\bh_{n+2}^{1/2}(a\otimes\id_{\Lambda_{[n+1,n+2]}})\bh_{n+2}^{1/2}K_{[n+1,n+2]}^{*})$ \\[2mm]
$=\tr(K_{[n,n+1]}^{*}K_{[n-1,n]}^{*}\cdots K_{[0,1]}^{*} w_0 K_{[0,1]}K_{[1,2]}
\cdots K_{[n,n+1]}(a\otimes\id_{\Lambda_{[n+1,n+2]}})K_{[n+1,n+2]}\bh_{n+2}K_{[n+1,n+2]}^{*})$\\[2mm]
$=\tr(K_{[n,n+1]}^{*}K_{[n-1,n]}^{*}\cdots K_{[0,1]}^{*} w_0 K_{[0,1]}K_{[1,2]}
\cdots K_{[n,n+1]}(a\otimes\id_{\Lambda_{[n+1,n+2]}})\Tr_{n+1]}(K_{[n+1,n+2]}\bh_{n+2}K_{[n+1,n+2]}^{*}))$\\[2mm]
$=\tr(K_{[n,n+1]}^{*}K_{[n-1,n]}^{*}\cdots K_{[0,1]}^{*} w_0 K_{[0,1]}K_{[1,2]}
\cdots K_{[n,n+1]}(a\otimes\id_{\Lambda_{[n+1,n+2]}})\bh_{n+1})$\\[2mm]
$=\tr(\bh_{n+1}^{1/2}K_{[n,n+1]}^{*}K_{[n-1,n]}^{*}\cdots K_{[0,1]}^{*} w_0 K_{[0,1]}K_{[1,2]}
\cdots K_{[n,n+1]}\bh_{n+1}^{1/2}(a\otimes\id_{\Lambda_{[n+1,n+2]}}))$\\[2mm]
$=\tr(\cw_{n+1]}(a\otimes\id_{\Lambda_{[n+1,n+2]}}))$\\[2mm]
$=\ffi^{(n)}_{w_0,{\bh}}(a)$

To  show that  $\ffi^{(n)}_{w_0,{\bh}}$  is a backward QMC, we define quasi-conditional
expectations $\ce_{n} $  as follows:
\begin{eqnarray}\label{E-n1}
&&\hat\ce_{n}(x_{n+1]})=\tr_{n]}(K_{[n,n+1]}\bh_{n+1}^{1/2}x_{n+1]}\bh_{n+1}^{1/2}K_{[n,n+1]}^*), \ \ x_{n+1]}\in \cb_{\L_{n+1}}\\
&&\label{E-n2}
\ce_{k}(x_{k+1]})=\tr_{k]}(K_{[k,k+1]}x_{k+1]}K_{[k,k+1]}^*),
\ \ x_{k+1]}\in\cb_{\L_{k+1}}, \ \ k=1,2,\dots,n-1,
\end{eqnarray}
Then for any monomial
$a_{\L_1}\otimes a_{W_2}\otimes\cdots\otimes
a_{W_{n}}\otimes\id_{W_{n+1}}$, where
$a_{\L_1}\in\cb_{\L_1},a_{W_k}\in\cb_{W_k}$, ($k=2,\dots,n$), we
have:\\[2mm]
$
\ffi^{(n)}_{w_0,\bh}(a_{\L_1}\otimes \cdots\otimes
a_{W_{n}})=\tr(\cw_{n+1]}(a_{\L_1}\otimes \cdots\otimes
a_{W_{n}}\otimes \id_{W_{n+1}}))$\\[2mm]
$=\tr(w_0 K_{[0,1]}
\cdots K_{[n,n+1]}\bh_{n+1}^{1/2}(a_{\L_1}\otimes \cdots\otimes
a_{W_{n}}\otimes \id_{W_{n+1}})\bh_{n+1}^{1/2}K_{[n,n+1]}^{*}\cdots K_{[0,1]}^{*})\\[2mm]
$$=\tr(w_0 K_{[0,1]} \cdots K_{[n-1,n]} \tr_{n]}(K_{[n,n+1]}\bh_{n+1}^{1/2}(a_{\L_1}\otimes \cdots\otimes
a_{W_{n}}\otimes \id_{W_{n+1}})\bh_{n+1}^{1/2}K_{[n,n+1]}^{*})K_{[n-1,n]}^{*}\cdots K_{[0,1]}^{*})\\[2mm]
$$=\tr(w_0 K_{[0,1]} \cdots K_{[n-2,n-1]} \tr_{n-1]}(K_{[n-1,n]} \hat\ce_{n}(a_{\L_1}\otimes \cdots\otimes
a_{W_{n}}) K_{[n-1,n]}^{*})\cdots K_{[0,1]}^{*})\\[2mm]
$$=\tr(w_0 K_{[0,1]} \cdots K_{[n-2,n-1]} \ce_{n-1} \hat\ce_{n}(a_{\L_1}\otimes \cdots\otimes
a_{W_{n}})K_{[n-2,n-1]}^{*}\cdots K_{[0,1]}^{*})\\[2mm]
$$=\tr(w_0 \ce_{0}\circ \ce_{1}\cdots \ce_{n-1}\circ \hat\ce_{n}(a_{\L_1}\otimes \cdots\otimes
a_{W_{n}}))$.
This means that the limit state $\ffi_{w_0,\bh}$ is a backward QMC. This completes the proof.
\end{proof}

We notice that a phase transition phenomena is crucial in
higher dimensional quantum models
\cite{BCS},\cite{Sach,BR2}. In \cite{ArE}, quantum phase transition for the two-dimensional Ising
model using $C^*$-algebra approach.  In \cite{fannes} the
VBS-model was considered on the Cayley tree. It was established
the existence of  the phase transition for the model in term of
finitely correlated states which describe ground states of the
model. Note that more general structure of finitely correlated
states was studied in \cite{fannes2}.

Our goal in this paper is to establish the existence of phase
transition for the given family of operators. Heuristically, the
phase transition means the existence of two distinct B
backward QMC. Let us
provide a more exact definition (see \cite{MBS161}).

\begin{definition}
We say that there exists a phase transition for a family of
operators $\{K_{<x,(x,i)>}\}$, $\{L_{>(x,i),(x,i+1)<}\}$,
$\{M_{(x,(x,i),(x,i+1))}\}$ if the following conditions are
satisfied:
\begin{enumerate}
\item[(a)] {\sc existence}: The equations \eqref{eq1}, \eqref{eq2}
have at least two $(u_0,\{h^x\}_{x\in L})$ and $(v_0,\{s^x\}_{x\in
L})$ solutions;

\item[(c)] {\sc not overlapping supports}: there is a projector
$P\in B_L$ such that $\ffi_{u_0,\bh}(P)<\ve$ and
$\ffi_{v_0,\bs}(P)>1-\ve$, for some $\ve>0$.

\item[(b)] {\sc not quasi-equivalence}: the corresponding quantum
Markov chains $\ffi_{u_0,\bh}$ and $\ffi_{v_0,\bs}$ are not quasi
equivalent \footnote{Recall that a representation $\pi_1$ of a
$C^*$-algebra $\ga$ is \textit{normal} w.r.t. another
representation $\pi_2$, if there is a normal $*$- epimorphism
$\rho:\pi_2(\ga)''\to \pi_1(\ga)''$ such that
$\rho\circ\pi_2=\pi_1$. Two representations $\pi_1$ and $\pi_2$
are called \textit{quasi-equivalent} if $\pi_1$ is normal w.r.t.
$\pi_2$, and conversely, $\pi_2$ is normal w.r.t. $\pi_1$. This
means that there is an isomorphism  $\alpha:\pi_1(\ga)''\to
\pi_2(\ga)''$ such that $\alpha \circ\pi_1=\pi_2$. Two states
 $\varphi$ and $\psi$ of $\ga$ are said be \textit{quasi-equivalent} if
 the GNS representations $\pi_\varphi$ and $\pi_\psi$ are
 quasi-equivalent \cite{BR}.}.

\end{enumerate}
Otherwise, we say there is no phase transition.
\end{definition}


\section{ QMC associated with Ising-XY model with competing interactions}\label{exam1}

In this section, we define the model and formulate the main
results of the paper. In what follows we consider a semi-infinite
Cayley tree $\G^2_+=(L,E)$ of order two. Our starting
$C^{*}$-algebra is the same $\cb_L$ but with $\cb_{x}=M_{2}(\bc)$
for all $x\in L$. By $\sigma_{x}^{u}$, $\sigma_{y}^{u}$, $\sigma_{z}^{u}$ we denote the Pauli spin operators for a site $u \in L$, i.e.

        $$\id^{(u)}=\begin{pmatrix}
0 & 1\\
 1& 0
\end{pmatrix},\sigma_{x}^{u}=\begin{pmatrix}
0 & 1\\
 1& 0
\end{pmatrix} , \sigma_{y}^{u}=\begin{pmatrix}
0 & -i\\
 i& 0
\end{pmatrix}, \sigma_{z}^{u}=\begin{pmatrix}
1 & 0\\
 0& -1
\end{pmatrix}.$$

For every vertices  $(u,(u,1),(u,2))$  we put
\begin{eqnarray}\label{1Kxy1}
&&K_{<u,(u,i)>}=\exp\{J_0\beta H_{<u,(u,i)>}\}, \ \ i=1,2,\ \ J_0>0, \ \beta>0,\\[2mm] \label{1Lxy1} &&
L_{>(u,1),(u,2)<}=\exp\{J\beta H_{>(u,1),(u,2)<}\}, \ \ J\in\br,
\end{eqnarray}
where
\begin{eqnarray}\label{1Hxy1}
&&
H_{<u,(u,1)>}=\frac{1}{2}\big(\id^{(u)}\otimes \id^{(u,1)}+\sigma_{z}^{(u)}\otimes \sigma_{z}^{(u,1)}\big)  ,\\
\label{1H>xy<1} &&
H_{>(u,1),(u,2)<}=\frac{1}{2}\big(\sigma_{x}^{(u,1)}\otimes \sigma_{x}^{(u,2)}+\sigma_{y}^{(u,1)}\otimes \sigma_{y}^{(u,2)}\big).
\end{eqnarray}

Furthermore, for the sake of simplicity, we assume that $M_{(u,(u,i),(u,i+1))}=\id$
($i=1,2,\dots,k$) for all $u\in L$.

The defined model is called  the {\it Ising model with competing $XY$-
interactions} per vertices  $(u,(u,1),(u,2))$.

For each $m \in \mathbb{N}$, from \eqref{1Hxy1}, \eqref{1H>xy<1} it follows that
\begin{eqnarray}\label{1Hxym}
&&H_{<u,v>}^{m}=H_{<u,v>}=\frac{1}{2}\big(\id^{(u)}\otimes\id^{(v)}+\sigma^{(u)}_z\otimes\sigma^{(v)}_z\big),\\[2mm]
\label{1Hxym}
&&H_{>(u,1),(u,2)<}^{2m}=H_{>(u,1),(u,2)<}^2=\frac{1}{2}(\id^{(u,1)}\otimes \id^{(u,2)}-\sigma_{z}^{(u,1)}\otimes \sigma_{z}^{(u,2)})\\[2mm]
\label{1Hxym}
&&H_{>(u,1),(u,2)<}^{2m-1}=H_{>(u,1),(u,2)<}\big.
\end{eqnarray}
Therefore, one finds
$$K_{<u,(u,i)>}=K_{0}\id^{(u)}\otimes\id^{(u,i)}+K_{3}\sigma_{z}^{(u)}\otimes \sigma_{z}^{(u,i)}   $$
$$L_{>(u,1),(u,2)<}= \id^{(u,1)}\otimes \id^{(u,2)} + \sinh(J\beta) H_{>(u,1),(u,2)<}+(\cosh(J\beta)-1)H^{2}_{>(u,1),(u,2)<}$$
where
\begin{eqnarray*}
&&K_0=\frac{\exp{J_0\beta}+1}{2},\ \ \   K_3=\frac{J_0\exp{\beta}-1}{2},\\[2mm]
\end{eqnarray*}
Hence, from \eqref{K1} for each $x\in L$ we obtain
\begin{eqnarray}\label{Ax}
A_{(u,(u,1),(u,2))}&=\gamma_{1}\id^{(u)}\otimes\id^{(u,1)}\otimes\id^{(u,2)}+\gamma_{2}\id^{u}
\otimes\sigma_{x}^{(u,1)}\otimes\sigma_{x}^{(u,2)}\nonumber\\
&+\gamma_{2}\id^{(u)}\otimes\sigma_{y}^{(u,1)}\otimes\sigma_{y}^{(u,2)}
+\gamma_3\id^{(u)}\otimes\sigma_{z}^{(u,1)}\otimes\sigma_{z}^{(u,2)}\nonumber\\
&+\delta_{1}\sigma_{z}^{(u)}\otimes\id^{(u,1)}\otimes\sigma_{z}^{(u,2)}
+\delta_1\sigma_{z}^{(u)}\otimes\sigma_{z}^{(u,1)}\otimes \id^{(u,2)}
\end{eqnarray}
where
$$\left\{\begin{matrix}
\gamma_1=\frac{1}{4}[\exp(2J_0\beta)+1+2\exp(J_0\beta)\cosh(J\beta)],  & \gamma_2=\frac{1}{2}\exp(J_0\beta)\sinh(J\beta). \bigskip \\
  \gamma_3=\frac{1}{4}[\exp(2J_0\beta)+1-2\exp(J_0\beta)\cosh(J\beta)],  & \delta_1=\frac{1}{4}(\exp(2J_0\beta)-1).
\end{matrix}\right.$$

Recall that a function $\{h^u\}$ is called
\textit{translation-invariant} if one has $h^{u}=h^{\tau_g(u)}$,
for all $u,g\in \G^2_+$. Clearly, this is equivalent to $h^u=h^v$
for all $u,v\in L$.

In what follows, we restrict ourselves to the description of
translation-invariant solutions of \eqref{eq1},\eqref{eq2}.
Consequently,  we assume that: $ h^{u}=h$ for all $u\in L$, here
$$
h= \left(
\begin{array}{cc}
            h_{11} & h_{12} \\
            h_{21} & h_{22} \\
          \end{array}
\right).$$
Then, equation \eqref{eq2} reduces to
\begin{eqnarray} \label{eqder1}
h&=&\Tr_{u]}A_{(u,(u,1),(u,2))}[\id^{(u)}\otimes h\otimes h]A_{(u,(u,1),(u,2))}^{*}\nonumber\\
&=&[C_1 \tr(h)^2+C_{2}\tr(\sigma_{z}h)^2]\id^{(u)}+ C_3Tr(h)\Tr(\sigma_z h)\sigma_z^{(u)}.
\end{eqnarray}
 where
 $$\left\{
  \begin{array}{ll}
    C_1=\frac{1}{4}(\exp(4 J_0\beta)+1)+\frac{1}{2}\exp(2 J_0\beta)\cosh(2J\beta);\bigskip \\
    C_2=\frac{1}{4}(\exp(4 J_0\beta)+1)-\frac{1}{2}\exp(2 J_0\beta)\cosh(2J\beta);\bigskip  \\
    C_3=\frac{1}{2}(\exp(4 J_0\beta)-1).
  \end{array}
\right.$$
Now taking into account
$$\Tr(h)=\frac{h_{11}+h_{22}}{2},\   \
\Tr(\sigma_{z} h)=\frac{h_{11}-h_{22}}{2}$$ the equation \eqref{eqder1}
reduces to the following one
\begin{equation}\label{1EQ1}
\left\{
   \begin{array}{lll}
\Tr(h)=C_1\Tr(h)^2+C_2\Tr(\s h)^2,\\
\Tr(\s h)=C_3\Tr(h)\Tr(\s h),\\
h_{21}=0, h_{12}=0.\\
   \end{array}
 \right.
 \end{equation}
This equation implies that a solution $h$ is diagonal, and through the equation \eqref{eq1},
$\omega_{0}$ could be chosen diagonal  as well.
In the next sections we are going to examine \eqref{1EQ1}.
\section{Existence of  QMC associated with the model.}
In this section we are going to solve \eqref{1EQ1}, which yields
the existence of QMC associated with the model. We consider two distinct cases.

\subsection{Case $h_{11}=h_{22}$ and associate QMC}

We assume that $h_{11}=h_{22}$ , then \eqref{1EQ1} is reduced to
\begin{equation*}
    h_{11}=h_{22}=\frac{1}{C_1}.
\end{equation*}
Then putting $\alpha=\frac{1}{C_1}$ we get
\begin{equation}
h_{\alpha}=
 \left(
  \begin{array}{cc}
    \alpha & 0 \\
    0 & \alpha \\
  \end{array}
\right)
\end{equation}
\begin{proposition}\label{5.1}
The pair $(\omega_{0},\{h^{u}=h_{\alpha}| u\in L\})$ with
$\omega_{0}=\frac{1}{\alpha}\id,\   \  h^{u}=h_{\alpha}, \forall
u\in L,$ is a solution of \eqref{eq1},\eqref{eq2}. Moreover the
associated backward QMC  can be written on the local algebra
$\mathcal{B}_{L, loc}$ by
\begin{equation}
\varphi_{\alpha}(a)=\alpha^{2^{n}-1}\Tr\bigg(\prod_{i=0}^{n-1}K_{[i,i+1]}a\prod_{i=0}^{n-1}K_{[n-i-1,n-i]}^{*}\bigg),
\ \ \forall a\in B_{\Lambda_{n}}.
\end{equation}
\end{proposition}

\subsection{Case  $h_{11}\neq h_{22}$ and associate QMC}

Now we suppose that $h_{11}\neq h_{22}$, and put $\theta=\exp(2\beta)$.
Then
the equation \eqref{1EQ1} reduces to
\begin{eqnarray}\label{EQ2}
\left\{
  \begin{array}{ll}\
    \frac{h_{11}+h_{22}}{2}=\frac{1}{C_3},  \\
    \\
        (\frac{h_{11}-h_{22}}{2})^{2}=\frac{C_{3}-C_1}{C_2.C_{3}^{2}},
  \end{array} \right.
\end{eqnarray}

Denote
\begin{equation}
\Delta(\theta)=\frac{C_{3}-C_1}{C_2}=\frac{\theta^{2J_0} -\theta^{J_0} (\theta^{J}+\theta^{-J})-3}{\theta^{2J_0}-\theta^{J_0} (\theta^{J}+\theta^{-J})+1}
\end{equation}

\begin{proposition}\label{hh0}
If $\Delta(\theta) >0$, then the equation \eqref{1EQ1} has two
solutions given by
\begin{eqnarray}\label{hh1}
  && h =\xi_{0}\id+\xi_{3}\sigma,\\
&&\label{hh2}
   h'=\xi_{0}\id-\xi_{3}\sigma,
\end{eqnarray}
where
\begin{eqnarray}\label{xi01}
\xi_{0}=\frac{1}{C_3}=\frac{2}{\theta^{2J_0}-1}, \ \ \
\xi_{3}=\frac{\sqrt{\Delta(\theta)}}{C_3}=\frac{2\sqrt{\Delta(\theta)}}{\theta^{2J_0}-1}
\end{eqnarray}
\end{proposition}
\begin{proof}
Assume that  $\Delta(\theta)>0$. Then one can conclude that
\eqref{EQ2} is equivalent to the following system
\begin{eqnarray*}
\left\{
  \begin{array}{ll}
    h_{1,1}+h_{2,2}=2\xi_0,  \\
      h_{1,1}-h_{2,2}=\pm2\xi_3
  \end{array}
\right.
\end{eqnarray*}
It is easy to see that $h_{1,1}=\xi_{0}-\xi_{3}$,
$h_{2,2}=\xi_{0}+\xi_3$. Hence, we get \eqref{hh1},\eqref{hh2}.
\end{proof}


%

From \eqref{eq1} we find that $\omega_0=\frac{1}{\xi_0}\id\in
\mathcal{B}^{+}$. Therefore, the pairs $\big(\omega_{0},\ \
\{h^{(u)}=h, \ u\in L\}\big)$ and $\big(\omega_0,\{h^{(u)}=h', \
u\in L\}\big)$ define two solutions of \eqref{eq1},\eqref{eq2}.
Hence, they define two backward QMC $\varphi_1$ and $\varphi_2$,
respectively. Namely, for every $ a\in\mathcal{B}_{\Lambda_n}$ one
has
\begin{eqnarray}\label{F1}
&&\varphi_1(a)=\Tr \big(\omega_0 K_{[0,1]}\cdots K_{[n-1,n]}\bh_{n}^{1/2} a  \bh_{n}^{1/2} K_{[n-1,n]}^{*}\cdots K_{[0,1]}^{*}\big)\\[2mm] \label{F2}
&&\varphi_2(a)=\Tr \big(\omega_0 K_{[0,1]}\cdots K_{[n-1,n]}\bh_{n}^{'1/2} a   \bh_{n}^{'1/2} K_{[n-1,n]}^{*}\cdots K_{[0,1]}^{*}\big).
\end{eqnarray}

Hence, we summarize this section in the followigin result.

\begin{theorem}\label{5.3}
The following statements hold:
\begin{itemize}
\item[(i)] if $\Delta(\theta)\leq0$, then there is a unique
translation invariant QMC $\varphi_\a$;
\item[(ii)] if $\Delta(\theta)>0$, then there are at least three
translation invariant QMC $\varphi_\a$, $\varphi_1$ and
$\varphi_2$.
\end{itemize}
\end{theorem}

From this theorem we conclude that there are three coexisting phases in the region  $\Delta(\theta)>0$, and one of it, i.e. $\varphi_{\alpha}$, survives in the region $\Delta(\theta)<0$. This leads us that
the state $\varphi_{\alpha}$ describes the disordered phase of the model, which shows a similar behavior with the classical Ising model \cite{B90,Roz2}.
In comparison with the Ising  model, we stress  that in the present model, we have a similar kind of phases (translation invariant ones) when  $\Delta(\theta)>0$. From Figure 2 (see below), one concludes that the phase transition
occurs except for a "triangular region". This shows how the competing interactions effect to the existence of other phases. Notice that if $J=0$, then we obtain the classical Ising model for which the existence of a disordered phase coexisting with
two ordered phases is well-known \cite{Bax, Roz2}.

Next auxiliary fact gives an equivalent condition for $\Delta(\theta)>0$.

\begin{lemma}\label{DD} $\Delta(\theta)>0$ iff one of the following statements hold:
\begin{enumerate}
\item[(i)] $J>J_0$ or $J<-J_0$;
\item[(ii)] $-J_0<J<J_0$ and
$$
J_0>\frac{1}{2\beta}\ln\left(\sinh(2J\beta)+\sqrt{\sinh^2(2J\beta)+3}\right)
$$
\end{enumerate}
\end{lemma}

\begin{proof}
We know that  $\theta=\exp(2\beta)$, $\beta>0$ and $J_0>0$. Then one finds

$$
\Delta=\frac{\theta^{2J_0}-\theta^{J_0}\left(\theta^J+\theta^{-J}\right)-3}
{\theta^{2J_0}-\theta^{J_0}\left(\theta^J+\theta^{-J}\right)+1}=\frac{R(J)-4}{R(J)}
$$
where
$$
R(J)=\theta^{2J_0}-\theta^{J_0}\left(\theta^J+\theta^{-J}\right)+1=(\theta^{J_0-J}-1)(\theta^{J_0+J}-1)
$$
Thanks to  $\theta>1$, we have
$$
R(J)\left\{
\begin{array}{ll}
>0, & \mbox{if } (J_0-J)(J_0+J)>0\\
<0, & \mbox{if } (J_0-J)(J_0+J)<0
\end{array}\right.
$$

{\bf Case $ (J_0-J)(J_0+J)<0$}. In this setting, one can see  that
$\Delta(\theta)>0$.\\

{\bf Case $ (J_0-J)(J_0+J)>0$} Note that $\Delta(\theta)>0$ if and only if $R(J)>4$.
For convenience, we denote $\theta^{J_0}=a$ and $\theta^J=b$. And consider the following
equation
$$
a^2-\left(b+b^{-1}\right)a-3=0.
$$
Then
$$
a_{1}=\frac{b+b^{-1}+\sqrt{D}}{2}>0
$$
$$
a_{2}=\frac{b+b^{-1}-\sqrt{D}}{2}<0
$$
where $D=\left(b+b^{-1}\right)^2+12$.
Due to $a>0$, we may conclude that
$$
a^2-\left(b+b^{-1}\right)a-3\left\{
\begin{array}{ll}
>0, & \mbox{if } a>a_1\\
<0, & \mbox{if } 0<a<a_1
\end{array}
\right.
$$
This means that $\Delta(\theta)>0$ if and only if
$$
J_0>\frac{1}{2\beta}\ln\left(\sinh(2J\beta)+\sqrt{\sinh^2(2J\beta)+3}\right)
$$
which completes the proof.
\end{proof}

From this lemma we infer that the phase transition exists in the shaded region shown in the Figure 2 (see $(J,J_0)$ plane).

\begin{figure}
\begin{center}
\includegraphics[width=10.07cm]{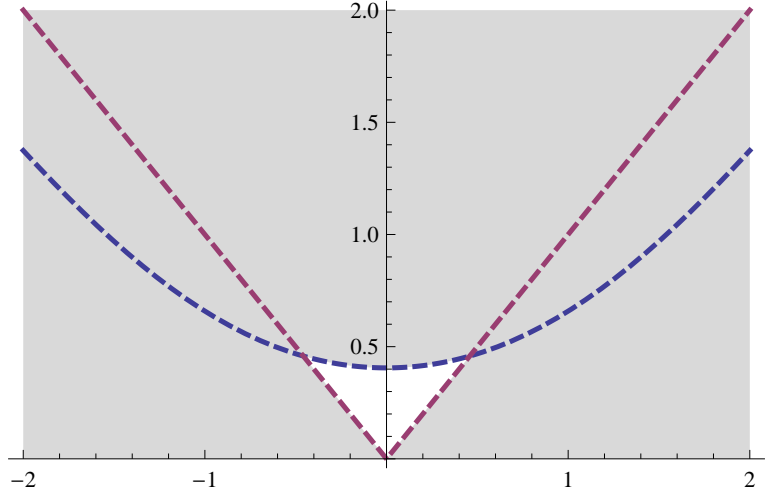}
\end{center}
\caption{Phase diagram} \label{fig1}
\end{figure}

\section{Proof of Theorem \ref{Main}}

This section is devoted to the proof of Theorem \ref{Main}. To
realize it, we first show not overlapping supports of the states
$\varphi_1$ and $\varphi_2$. Then we show that these states satisfy the clustering property, which yields that they are factor states, and this fact allows us to prove their  non-
quasi-equivalence.

\subsection{Not overlapping supports of $\varphi_1$ and $\varphi_2$}
As usual we put
$$
e_{11}= \left(
\begin{array}{cc}
                  1 & 0 \\
                  0 & 0
                \end{array}
                \right),  \    \           \
                e_{22}=
                \left( \begin{array}{cc}
                  0 & 0 \\
                  0 & 1
                \end{array}
                \right).
$$
Now for each $n\in\mathbf{N}$, we denote
$$
P_n:=\bigg(\bigotimes_{x\in \Lambda_n}e_{11}^{(x)}\bigg)\otimes\id, \             \ Q_n:=\bigg(\bigotimes_{x\in
\Lambda_n}e_{22}^{(x)}\bigg)\otimes\id.
$$
Clearly, $P_n$ and $Q_n$ are orthogonal projections in $\mathcal{B}_{\Lambda_n}$.

\begin{lemma}\label{pq_n} For every $n\in\bn$, one has
\begin{enumerate}
\item[(i)] $\varphi_{1}(P_n)=\varphi_2(Q_n)=\frac{1}{2\xi_0}\left(\xi_0+\xi_3\right)^{2^{n}}\left(\frac{C_1+C_2+C_3}{4}\right)^{2^{n}-1},$
\item[(ii)] $\varphi_1(Q_n)=\varphi_2(P_n)=\frac{1}{2\xi_0}\left(\xi_0-\xi_3\right)^{2^{n}}\left(\frac{C_1+C_2+C_3}{4}\right)^{2^{n}-1}.$
\end{enumerate}
\end{lemma}
\begin{proof} (i). From \eqref{F1} we find
\begin{eqnarray*}
\varphi_1(P_n)&=&\Tr \big(\omega_0 K_{[0,1]}\cdots K_{[n-1,n]}\bh_{n}^{1/2} P_n  \bh_{n}^{1/2} K_{[n-1,n]}^{*}\cdots K_{[0,1]}^{*}\big)\nonumber\\
&=&\Tr \big(\omega_0 K_{[0,1]}\cdots K_{[n-1,n]}\bh_{n} P_n K_{[n-1,n]}^{*}\cdots K_{[0,1]}^{*}\big)
\end{eqnarray*}
Thanks to $he_{11}=(\xi_0+\xi_3)e_{11}$ and
\eqref{Ax} one gets
\begin{eqnarray*} \Tr_{n-1]}(
K_{[n-1,n]}\bh_nP_nK_{[n-1,n]}^{*})
&=& P_{n-2}\otimes\prod_{u\in  \overrightarrow W_{n-1}}\Tr_{u]}(A_{(u,(u,1),(u,2))}e_{11}^{(u)}\otimes h e_{11}^{(u,1)}\otimes h
e_{11}^{(u,2)}A_{(u,(u,1),(u,2))})\nonumber\\
&=& (\xi_0+\xi_3)^{2|W_{n-1]}|}\bigg(\frac{C_1+C_2+C_3}{4}\bigg)^{|W_{n-1}|}P_{n-1}.\nonumber\\
\end{eqnarray*}
Hence,
\begin{eqnarray*}
\varphi_1(P_n)&=&(\xi_0+\xi_3)^{|W_{n}|}\left(\frac{C_1+C_2+C_3}{4}\right)^{|W_{n-1}|}\Tr\left[ \omega_0K_{[0,1]}\cdots
K_{[n-2,n-1]}P_{n-1}K_{[n-2,n-1]}^{*}\cdots K_{[0,1]}^{*}\right]\nonumber\\
&&.\nonumber\\
&&.\nonumber\\
&&.\nonumber\\
&=& \left(\xi_0+\xi_3\right)^{|W_{n}|}\left(\frac{C_1+C_2+C_3}{4}\right)^{|W_{n-1}|+ ... +|W_{0}|}\Tr\left[ \omega_0P_{0}\right]\nonumber\\
&=&\left(\xi_0+\xi_3\right)^{|W_{n}|}\left(\frac{C_1+C_2+C_3}{4}\right)^{|\Lambda_{n-1}|}\Tr\left[ \omega_0P_{0}\right]\nonumber\\
&=&\frac{1}{2\xi_0}\left(\xi_0+\xi_3\right)^{2^{n}}\left(\frac{C_1+C_2+C_3}{4}\right)^{2^n-1}.
\end{eqnarray*}
Analogously, using  $h'e_{22}=(\xi_0+\xi_3)e_{22}$
we obtain
\begin{eqnarray*}
\Tr_{n-1]} (K_{[n-1,n]}\bh'_nQ_nK_{[n-1,n]}^{*})
&=& Q_{n-2}\otimes\prod_{u\in \overrightarrow W_{n-1}}\Tr_{u]}(A_{(u,(u,1),(u,2))}e_{2,2}^{(u)}\otimes h' e_{2,2}^{(u,1)}\otimes h'
e_{22}^{(u,2)}A_{(u,(u,1),(u,2))})\nonumber\\
&=&
(\xi_0+\xi_3)^{2|\mathcal{W}_{n-1]}|}\bigg(\frac{C_1+C_2+C_3}{4}\bigg)^{|\mathcal{W}_{n-1]}|}Q_{n-1}.
\end{eqnarray*}
which yields
$$\varphi_2(Q_n)=\frac{1}{2\xi_0}\left(\xi_0+\xi_3\right)^{2^{n}}\left(\frac{C_1+C_2+C_3}{4}\right)^{2^n-1}.$$

(ii) Now from $he_{22}=(\xi_0-\xi_3)e_{22}$ and
$h'e_{11}=(\xi_0-\xi_3)e_{11}$, we obtain
\begin{eqnarray*}
&&\Tr_{n-1]} K_{[n-1,n]}\bh_nQ_nK_{[n-1,n]}^{*} =
(\xi_0-\xi_3)^{2|W_{n-1}|}\bigg(\frac{C_1+C_2+C_3}{4}\bigg)^{|W_{n-1}|}Q_{n-1}.\\[2mm]
&&\Tr_{n-1]} K_{[n-1,n]}\bh'_nP_nK_{[n-1,n]}^{*} =
(\xi_0-\xi_3)^{2|W_{n-1}|}\bigg(\frac{C_1+C_2+C_3}{4}\bigg)^{|W_{n-1}|}P_{n-1}.
\end{eqnarray*}
The same argument as above implies (ii). This completes the proof.
\end{proof}

\begin{theorem}\label{6.2}
For fixed $n\in\mathbf{N}$, one has
$$\varphi_1(P_n) \rightarrow 1, \      \   \varphi_2(Q_n)\rightarrow 1 \ \ \textrm{as} \ \ \beta\rightarrow +\infty.
$$
\end{theorem}
\begin{proof} We know that $\theta=\exp (2\beta)\rightarrow+\infty $ as
$\beta\rightarrow+\infty$. Hence, one finds
\begin{eqnarray*}
&& \frac{1}{2\xi_0}=\frac{\theta^{2J_0}-1}{4}\sim
\frac{\theta^{2J_0}}{4}, \ \ \textrm{as} \ \
\theta\to+\infty\\[2mm]
&&
\left(\xi_0+\xi_3\right)^{2^{n}}=\left(\frac{2}{\theta^{2J_0}-1}(1+\sqrt{\frac{\theta^{2J_0}-\theta^{J_0}(\theta^J+\theta^{-J})+1}{\theta^{2J_0}-\theta^{J_0}(\theta^J+\theta^{-J})-3}})\right)^{2^{n}}
\sim \left(\frac{4}{\theta^{2J_0}}\right)^{2^{n}}, \ \ \textrm{as}
\ \
\theta\to+\infty\\[2mm]
&&\left(\frac{C_1+C_2+C_3}{4}\right)^{2^{n}-1}=\left(\frac{\theta^{2J_0}}{4}\right)^{2^{n}-1}.
\end{eqnarray*}
Hence, we obtain
$$\varphi_{1}(P_n)=\varphi_2(Q_n)
\sim
\frac{\theta^{2J_0}}{4}\left(\frac{4}{\theta^{2J_0}}\right)^{2^{n}}\left(\frac{\theta^{2J_0}}{4}\right)^{2^{n}-1}=1.
$$
So, $$ \lim_{\theta\to\infty}\varphi_{1}(P_n)=
\lim_{\theta\to\infty} \varphi_2(Q_n)=1.$$ This completes the
proof.
\end{proof}
\begin{remark} We note that from  $P_n\leq 1-Q_n$ one gets
$$
\lim_{\theta\to\infty}\varphi_{1}(Q_n)= \lim_{\theta\to\infty}
\varphi_2(P_n)=0.$$
This implies that the states $\varphi_1$ and $\varphi_2$ have non
overlapping supports.
\end{remark}

\subsection{Clustering Property for $\varphi_1$ and $\varphi_2$ }

In this subsection, we are going to prove that the states $\ffi_1$, $\ffi_2$ satisfy the clustering property.

Recall that a state $\ffi$ on $\mathcal{B}_L$ satisfies the \textit{clustering property} if  for every $a,f\in \mathcal{B}_L$ one has
\begin{equation}\label{clus}
\lim_{|g|\rightarrow \infty}\ffi(a \tau_{g}(f))=\ffi(a)\ffi(f).
\end{equation}

Thanks to Theorem \eqref{5.3} there are two solutions of \eqref{eq1}, and  \eqref{eq2}, these two solutions  can be written as follows: $\big(\omega_{0},\ \ \{h^{(u)}=h, \ u\in L\}\big)$ and $\big(\omega_0,\{h^{(u)}=h', \
u\in L\}\big)$, where\\
$$\omega_{0}=\frac{1}{\xi_{0}}\id ,~~~h=\xi_{0}\id -\xi_{3}\sigma_{z},$$

$$\omega_{0}=\frac{1}{\xi_{0}}\id ,~~~h'=\xi_{0}\id +\xi_{3}\sigma_{z}.$$

By $\ffi_1$, $\ffi_2$ we denote the corresponding backward quantum
Markov chains. To prove the clustering property we need to study the following matrix:

 $$A:=\begin{pmatrix}
 c_{1}\xi_{0} & -c_{2}\xi_{3}\\
 -\frac{c_{3}}{2}\xi_{3} & \frac{c_{3}}{2}\xi_{0}
\end{pmatrix}$$

One can easily prove the following fact.

\begin{proposition}
The above given matrix $A$ is a diagonalizable matrix, and  can be written as follows:
\begin{equation}\label{7.2}
A:=P\begin{pmatrix}
 1 & 0\\
 0 & (c_{1}-\frac{c_{3}}{2})\xi_{0}
\end{pmatrix}P^{-1}
\end{equation}
where
 $$P:=\begin{pmatrix}
 \frac{c_{2}\xi_{3}}{c_{1}\xi_{0}-1} & 2c_{2}\xi_{3}\\
 1 & 1
\end{pmatrix}, \ \ \det(P)=\frac{c_{2}\xi_{3}(3-2c_{1}\xi_{0})}{c_{1}\xi_{0}-1}.$$
\end{proposition}

\begin{lemma}\label{7.3}
Let  $a \in \mathcal{B}_{\L_{N_{0}}}$, for some $N_{0} \in \bn$, and $f_n=\bigotimes\limits_{x\in W_{n}}f^{(x)}=f^{x_{W_{n}}^{(1)}} \otimes\id_{W_{n}\setminus\{x_{W_{n}}^{(1)}\}}\in\mathcal{B}_{W_{n}}$,  where $f^{x_{W_{n}}^{(1)}} =f$, then for each  backward quantum Markov chains $\ffi_1$, $\ffi_2$ we have
$$\lim_{n\rightarrow \infty}\ffi_k(a \otimes \id \cdots\otimes \id \otimes f_n)=\ffi_k(a)\ffi_k(f ), \ \ k=1,2.$$

\end{lemma}

\begin{proof}

 By symmetry of calculations, it is enough to prove the result for  $\big(\omega_{0},\ \ \{h^{(u)}=h, \ u\in L\}\big)$.
 From  \eqref{ffi-ff} and \eqref{E-n1} we have:

$$ \ffi^{(n)}_{w_{0},\bh}(a_{N_{0}} \otimes \id \cdots\otimes \id \otimes f)
=\tr(\omega_{0}\ce_{0}\circ \cdots\circ \ce_{N_{0}}( a\otimes \ce_{N_{0}+1}( \id \otimes \cdots\otimes \ce_{n-1}( \id \otimes\hat\ce_{n}(\id \otimes f))\cdots))$$
First, let us calculate $\hat\ce_{n}(\id \otimes f)$. From \eqref{E-n1} it follows that

 \begin{eqnarray*}
\hat\ce_{n}(f\otimes \id)&=& \bigotimes\limits_{x\in W_{n}}\tr_{x]}(A_{x \vee S(x)} f^{(x)}\otimes h_{S(x)} A^{*}_{x \vee S(x)})\\
 &=& \tr_{x_{W_{n}}^{(1)}]}(A_{x_{W_{n}}^{(1)} \vee S(x_{W_{n}}^{(1)})} f^{(x_{W_{n}}^{(1)})}\otimes h_{S(x_{W_{n}}^{(1)})} A^{*}_{x_{W_{n}}^{(1)} \vee S(x_{W_{n}}^{(1)})}) \otimes\\[2mm]
 &&\bigotimes\limits_{x\in W_{n}\setminus\{x_{W_{n}}^{(1)}\}} \tr_{x]}(  A_{x \vee S(x)}h_{S(x)}
   A^{*}_{x \vee S(x)})\\
   &=& \bigg( \alpha_{1} f^{(x_{W_{n}}^{(1)})}+  \alpha_{2} (f^{(x_{W_{n}}^{(1)})}\sigma_{z}^{(x_{W_{n}}^{(1)})}
   + \sigma_{z}^{(x_{W_{n}}^{(1)})}f^{(x_{W_{n}}^{(1)})})\\[2mm]
   &&+ \alpha_{3}\sigma_{z}^{(x_{W_{n}}^{(1)})}f^{(x_{W_{n}}^{(1)})}\sigma_{z}^{(x_{W_{n}}^{(1)})}\bigg)\otimes \bigotimes\limits_{x\in W_{n}\setminus\{x_{W_{n}}^{(1)}\}}h_{x}\\
   &=& g^{(x_{W_{n}}^{(1)})}\otimes \bigotimes\limits_{x\in W_{n}\setminus\{x_{W_{n}}^{(1)}\}}h_{x},
 \end{eqnarray*}
where  $$ g^{(x_{W_{n}}^{(1)})}=\alpha_{1} f^{(x_{W_{n}}^{(1)})}+  \alpha_{2} (f_{1}^{(x_{W_{n}}^{(1)})}\sigma_{z}^{(x_{W_{n}}^{(1)})}+ \sigma_{z}^{(x_{W_{n}}^{(1)})}f^{(x_{W_{n}}^{(1)})})+ \alpha_{3}\sigma_{z}^{(x_{W_{n}}^{(1)})}f^{(x_{W_{n}}^{(1)})}\sigma_{z}^{(x_{W_{n}}^{(1)})},$$
 and $$\left\{\begin{array}{ll}
 \alpha_{1}=(C_{1}-2\delta_{1}^{2})\xi_{0}^{2}+ (C_{2}-2\delta_{1}^{2})\xi_{3}^{2}\\
 \alpha_{2}=-\frac{C_{3}}{2}\xi_{0}\xi_{3}\\
 \alpha_{3}=2\delta_{1}^{2}(\xi_{0}^{2}+\xi_{3}^{2})
\end{array}\right..$$

Hence, one has
\begin{eqnarray*}
\ce_{n-1}(\id \otimes \hat\ce_{n}(f\otimes \id))&=& \tr_{x_{W_{n-1}}^{(1)}]}(A_{x_{W_{n-1}}^{(1)} \vee S(x_{W_{n-1}}^{(1)})} \id \otimes g^{(x_{W_{n}}^{(1)})}\otimes h A^{*}_{x_{W_{n-1}}^{(1)} \vee S(x_{W_{n-1}}^{(1)})})
\otimes \\[2mm]
&&\bigotimes\limits_{x\in W_{n-1}\setminus\{x_{W_{n-1}}^{(1)}\}} h_{(x)}\\
&=&(\alpha_{1,g}\id^{(x_{W_{n-1}}^{(1)})}+\alpha_{2,g}\sigma_{z}^{(x_{W_{n-1}}^{(1)})})\otimes \bigotimes\limits_{x\in W_{n-1}\setminus\{x_{W_{n-1}}^{(1)}\}} h_{(x)}\\
&=&\alpha_{1,g} \id^{(x_{W_{n-1}}^{(1)})}\otimes \bigotimes\limits_{x\in W_{n-1}\setminus\{x_{W_{n-1}}^{(1)}\}} h_{(x)}\\[2mm]
&&+ \alpha_{2,g} \sigma_{z}^{(x_{W_{n-1}}^{(1)})}\otimes \bigotimes\limits_{x\in W_{n-1}\setminus\{x_{W_{n-1}}^{(1)}\}} h_{(x)}
 \end{eqnarray*}

So, one finds \begin{eqnarray*}\ce_{n-1}(\id \otimes \hat\ce_{n}(f\otimes \id))&=& v_{1}\id^{(x_{W_{n-1}}^{(1)})}\otimes \bigotimes\limits_{x\in W_{n-1}\setminus\{x_{W_{n-1}}^{(1)}\}} h_{(x)}+ v^{'}_{1} \sigma_{z}^{(x_{W_{n-1}}^{(1)})}\otimes \bigotimes\limits_{x\in W_{n-1}\setminus\{x_{W_{n-1}}^{(1)}\}} h_{(x)}
\end{eqnarray*}
where $$\left\{\begin{array}{ll}
 \alpha_{1,X}=C_{1}\tr(X)\xi_{0}- C_{2}\tr(\sigma_{z}X)\xi_{3}\\
 \alpha_{2,X}=\frac{C_{3}}{2}\left(\tr(\sigma_{z}X)\xi_{0}-\tr(X)\xi_{3}\right)\\
 v_{1}=\alpha_{1,g}\\
 v_{1}^{'}=\alpha_{2,g}\\
\end{array}\right.$$

Then by iteration we obtain
\begin{eqnarray*}\ce_{n-k}(\id \otimes... \ce_{n-1}(\id \otimes \hat\ce_{n}(f\otimes \id)))&=& v_{k}\id^{(x_{W_{n-1}}^{(1)})}\otimes \bigotimes\limits_{x\in W_{n-1}\setminus\{x_{W_{n-1}}^{(1)}\}} h_{(x)}\\[2mm]
&&+ v^{'}_{k} \sigma_{z}^{(x_{W_{n-1}}^{(1)})}\otimes \bigotimes\limits_{x\in W_{n-1}\setminus\{x_{W_{n-1}}^{(1)}\}} h_{(x)}
\end{eqnarray*}
where $$\left\{\begin{array}{ll}
 v_{k}=v_{k-1}C_{1}\xi_{0}-C_{2}\xi_{3}v_{k-1}^{'}\\
 v_{k}^{'}=-\frac{C_{3}}{2}\xi_{3}v_{k-1}+ \frac{C_{3}}{2}\xi_{0}v_{k-1}^{'}\\
\end{array}\right.$$

Now let calculate the explicit form of the sequence $v_{k}$, we can see :
\begin{eqnarray*}\begin{pmatrix}
 v_{k}\\ v_{k}^{'}
\end{pmatrix}&=& A \begin{pmatrix}
v_{k-1}\\
 v_{k-1}^{'}
\end{pmatrix}\\
&\vdots&\\
&=& A^{k-1}\begin{pmatrix}
v_{1}\\
 v_{1}^{'}
\end{pmatrix}
\end{eqnarray*}

Then by \eqref{7.2} we get, \begin{eqnarray*}\begin{pmatrix}
 v_{k}\\ v_{k}^{'}
\end{pmatrix}&=&P\begin{pmatrix}
 1 & 0\\
0& (c_{1}-\frac{c_{3}}{2})^{k-1}\xi_{0}^{k-1}
\end{pmatrix}P^{-1}\begin{pmatrix}
 v_{1}\\
 v_{1}^{'}
\end{pmatrix}\\
&=&\begin{pmatrix}
 \frac{1}{3-2c_{1}\xi_{0}}+\frac{2\xi_{0}^{k-1}(c_{1}-\frac{c_{3}}{2})^{k-1}(1-c_{1}\xi_{0})}{3-2c_{1}\xi_{0}} & -\frac{2c_{2}\xi_{3}}{3-2c_{1}\xi_{0}}+\frac{2c_{2}\xi_{3}\xi_{0}^{k-1}(c_{1}-\frac{c_{3}}{2})^{k-1}}{3-2c_{1}\xi_{0}}\\[2mm]
\frac{c_{1}\xi_{0}-1}{c_{2}\xi_{3}(3-2c_{1}\xi_{0})}+\frac{\xi_{0}^{k-1}(c_{1}-\frac{c_{3}}{2})^{k-1}(1-c_{1}\xi_{0})}{c_{2}\xi_{3}(3-2c_{1}\xi_{0})}&
\frac{-2(c_{1}\xi_{0}-1)}{3-2c_{1}\xi_{0}}+\frac{\xi_{0}^{k-1}(c_{1}-\frac{c_{3}}{2})^{k-1}}{3-2c_{1}\xi_{0}}
\end{pmatrix}\begin{pmatrix}
 v_{1}\\[2mm]
 v_{1}^{'}
\end{pmatrix}\\[2mm]
&=&\begin{pmatrix}
 \eta_{1}+\widehat{\eta_{1}} \xi_{0}^{k-1}(c_{1}-\frac{c_{3}}{2})^{k-1} & \eta_{2}- \eta_{2}\xi_{0}^{k-1}(c_{1}-\frac{c_{3}}{2})^{k-1}\\[2mm]
 \widehat{\eta_{2}}-\widehat{\eta_{2}} \xi_{0}^{k-1}(c_{1}-\frac{c_{3}}{2})^{k-1}&
\widehat{\eta_{1}} + \eta_{1}\xi_{0}^{k-1}(c_{1}-\frac{c_{3}}{2})^{k-1}
\end{pmatrix}\begin{pmatrix}
 v_{1}\\[2mm]
 v_{1}^{'}\\
\end{pmatrix}
\end{eqnarray*}
where
 $$
\eta_{1}=\frac{1}{3-2c_{1}\xi_{0}},\ \
 \widehat{\eta_{1}}=\frac{2(1-c_{1}\xi_{0})}{3-2c_{1}\xi_{0}}
$$$$
\eta_{2}=-\frac{2c_{2}\xi_{3}}{3-2c_{1}\xi_{0}}, \ \
 \widehat{\eta_{2}}=\frac{c_{1}\xi_{0}-1}{c_{2}\xi_{3}(3-2c_{1}\xi_{0})}
$$

Hence,

$$\left\{\begin{array}{ll}
 v_{k}=\left(\eta_{1}+\widehat{\eta_{1}} \xi_{0}^{k-1}(c_{1}-\frac{c_{3}}{2})^{k-1} \right)v_{1}+ \left( \eta_{2}- \eta_{2}\xi_{0}^{k-1}(c_{1}-\frac{c_{3}}{2})^{k-1} \right)v_{1}^{'}\\[2mm]
 v_{k}^{'}=\left(\widehat{\eta_{2}}-\widehat{\eta_{2}} \xi_{0}^{k-1}(c_{1}-\frac{c_{3}}{2})^{k-1}  \right)v_{1}+ \left( \widehat{\eta_{1}} + \eta_{1}\xi_{0}^{k-1}(c_{1}-\frac{c_{3}}{2})^{k-1} \right)v_{1}^{'}\\
\end{array}\right.$$

So, one finds

$$\ffi^{(n)}_{w_{0},\bh}(a\otimes \id \cdots\otimes \id \otimes f)
=v_{n-N_{0}+1}\tr\left(\omega_{0}\ce_{0}\circ \cdots\circ \ce_{N_{0}}\bigg( a\otimes \id^{(x_{W_{N_{0}+1}})} \otimes \bigotimes\limits_{x\in W_{N_{0}+1}\setminus\{x_{W_{N_{0}+1}}^{(1)}\}} h_{x} \bigg)\right)$$
$$+ v_{n-N_{0}+1}^{'}\tr\left(\omega_{0}\ce_{0}\circ \cdots\circ \ce_{N_{0}}\bigg( a\otimes \sigma_{z}^{(x_{W_{N_{0}+1}})} \otimes \bigotimes\limits_{x\in W_{N_{0}+1}\setminus\{x_{W_{N_{0}+1}}^{(1)}\}} h_{x} \bigg)\right)$$
$$=\left(\eta_{1}+\widehat{\eta_{1}} \xi_{0}^{n-N_{0}}(c_{1}-\frac{c_{3}}{2})^{n-N_{0}} \right)v_{1}\tr\left(\omega_{0}\ce_{0}\circ\cdots\circ \ce_{N_{0}}\bigg( a\otimes \id^{(x_{W_{N_{0}+1}})} \otimes \bigotimes\limits_{x\in W_{N_{0}+1}\setminus\{x_{W_{N_{0}+1}}^{(1)}\}} h_{x} \bigg)\right)$$
$$+ \left( \eta_{2}- \eta_{2}\xi_{0}^{n-N_{0}}(c_{1}-\frac{c_{3}}{2})^{n-N_{0}} \right)v_{1}^{'}\tr\left(\omega_{0}\ce_{0}\circ \cdots\circ \ce_{N_{0}}\bigg( a\otimes \id^{(x_{W_{N_{0}+1}})} \otimes \bigotimes\limits_{x\in W_{N_{0}+1}\setminus\{x_{W_{N_{0}+1}}^{(1)}\}} h_{x} \bigg)\right)$$
$$+ \left(\widehat{\eta_{2}}-\widehat{\eta_{2}} \xi_{0}^{n-N_{0}}(c_{1}-\frac{c_{3}}{2})^{k-1}  \right)v_{1}\tr\left(\omega_{0}\ce_{0}\circ \cdots\circ \ce_{N_{0}}\bigg( a\otimes \sigma_{z}^{(x_{W_{N_{0}+1}})} \otimes \bigotimes\limits_{x\in W_{N_{0}+1}\setminus\{x_{W_{N_{0}+1}}^{(1)}\}} h_{x}\bigg )\right)$$
$$+ \left( \widehat{\eta_{1}} + \eta_{1}\xi_{0}^{n-N_{0}}(c_{1}-\frac{c_{3}}{2})^{n-N_{0}} \right)v_{1}^{'}\tr\left(\omega_{0}\ce_{0}\circ \cdots\circ \ce_{N_{0}}\bigg( a\otimes \sigma_{z}^{(x_{W_{N_{0}+1}})} \otimes \bigotimes\limits_{x\in W_{N_{0}+1}\setminus\{x_{W_{N_{0}+1}}^{(1)}\}} h_{x} \bigg)\right)
$$
 One can see that
$\xi_{0}^{n-N_{0}}\rightarrow 0$, as  $n\rightarrow \infty$, which implies

$$\lim_{n\rightarrow \infty}\ffi_{w_{0},\bh}(a \otimes \id \otimes \cdots\otimes \id \otimes f_n)
=(\eta_{1}v_{1}+ \eta_{2}v_{1}^{'})\tr\bigg(\omega_{0}\ce_{0}\circ \cdots\circ \ce_{N_{0}}( a\otimes \id^{(x_{W_{N_{0}+1}})} \otimes
$$
$$\bigotimes\limits_{x\in W_{N_{0}+1}\setminus\{x_{W_{N_{0}+1}}^{(1)}\}} h_{x} )\bigg)$$
$$+(\widehat{\eta_{2}}v_{1}+\widehat{\eta_{1}}v_{1}^{'})\tr\left(\omega_{0}\ce_{0}\circ \cdots\circ \ce_{N_{0}}( a\otimes \sigma_{z}^{(x_{W_{N_{0}+1}})} \otimes \bigotimes\limits_{x\in W_{N_{0}+1}\setminus\{x_{W_{N_{0}+1}}^{(1)}\}} h_{x} )\right)$$
where
$$\left\{\begin{array}{ll}
 \eta_{1}v_{1}+ \eta_{2}v_{1}^{'}=\frac{\xi_{0}^{2}}{6-4C_{1}\xi_{0}}[(4C_{3}-2C_{1})\tr(f)-\sqrt{\Delta(\theta)}(4C_{2}+C_{3})\tr(\sigma_{z}f)]\\[2mm]
 \widehat{\eta_{2}}v_{1}+\widehat{\eta_{1}}v_{1}^{'}=-c_{3}\xi_{3}(\eta_{1}v_{1}+ \eta_{2}v_{1}^{'})\\
\end{array}\right.$$

On other hand we have $$\ffi_{w_{0},\bh}(f)=C_{3}(\eta_{1}v_{1}+ \eta_{2}v_{1}^{'}),$$

Hence, one gets

\begin{eqnarray*}\lim_{n\rightarrow \infty}\ffi_{w_{0},\bh}(a\otimes\cdots\otimes f_n)&=&\ffi_{w_{0},\bh}(f)\tr\bigg(\omega_{0}\ce_{0}\circ \cdots\circ \ce_{N_{0}}( a\otimes h^{(x_{W_{N_{0}+1}})}\otimes \bigotimes\limits_{x\in W_{N_{0}+1}\setminus\{x_{W_{N_{0}+1}}^{(1)}\}} h_{x}\bigg)\\
&=&\ffi_{w_{0},\bh}(f)\tr\left(\omega_{0}\ce_{0}\circ \cdots\circ \hat\ce_{N_{0}}( a)\right)\\
&=&\ffi_{w_{0},\bh}(f)\ffi_{w_{0},\bh}(a).\end{eqnarray*} This completes the prove.
\end{proof}

Now we are ready to prove the clustering property.

\begin{theorem}\label{CPr}  The states $\ffi_1$ and $\ffi_2$ satisfy the clustering property.
\end{theorem}

\begin{proof}  Thanks to the density argument, without lost of generality, we may assume  $a,f \in \mathcal{B}_{loc}$. This means that
there are $N_0,m_{0} \in \bn$ such that  $a\in \mathcal{B}_{\L_{N_{0}}}$, $f \in \mathcal{B}_{\L_{m_{0}}}$. Moreover, $f$ can be write in the following form $$f=\bigotimes\limits_{x\in\L_{m_{0}}}f^{(x)}.$$
 By symmetry of calculations, it is enough to prove the result for  $\big(\omega_{0},\ \ \{h^{(u)}=h, \ u\in L\}\big)$.\\[2mm]
In what follows, we assume that $g\in W_n$. Therefore, we put $g_n:=g$. Then one has
\begin{eqnarray*}
\tau_{g_{n}}(f)&=& \bigotimes\limits_{x\in\L_{m_{0}}}f^{( g_{n} \circ x)}\\
&=& f^{(g_{n})}\otimes f^{(g_{n},W_{1})}\otimes f^{(g_{n},W_{2})}\otimes...\otimes f^{(g_{n},W_{m_{0}})},
\end{eqnarray*}
where  $$f^{(g_{n},W_{k})}= \bigotimes\limits_{x\in W_{k}}f^{(g_{n}\circ x)} \ \ \text{and} \ \ \{g_{n},W_{k}\}=\{(g_{n} \circ x),\ \ x \in W_{k}\}.$$\\

We can see $\tau_{g_{n}}(f) $ as an element of $B_{\L_{n+m_{0}}}$, i.e.

\begin{eqnarray*}
\tau_{g_{n}}(f)&=& \id_{\L_{n-1}}\otimes( f^{(g_{n})}\otimes \id_{W_{n}\setminus\{g_{n}\}})\otimes(f^{(g_{n},W_{1})}\otimes \id_{W_{n+1}\setminus\{g_{n},W_{1}\}})\otimes\\[2mm]
&&...\otimes(f^{(g_{n},W_{m_{0}})}\otimes \id_{W_{n+m_{0}}\setminus\{g_{n},W_{m_{0}}\}}).
\end{eqnarray*}

For the sake of simplicity, let us denote
\begin{eqnarray*}
\tau_{g_{n}}(f)&=& \bigotimes\limits_{x\in\L_{n+m_{0}}}f_{1}^{(x)}.
\end{eqnarray*}

From   \eqref{ffi-ff} and \eqref{E-n1} it follows that
 $$
  \ffi_{w_{0},\bh}(a\otimes \id \otimes \cdots\otimes \id \otimes \tau_{g_{n}}(f))=\ffi^{(n+m_{0})}_{w_{0},\bh}(a \otimes \id \cdots\otimes \id \otimes \tau_{g_{n}}(f))$$
$$=\tr(\omega_{0}\ce_{0}\circ\cdots\circ \ce_{N_{0}}( a\otimes \ce_{N_{0}+1}( \id \otimes \cdots\otimes \ce_{n+m_{0}-1}( \id \otimes\hat\ce_{n+m_{0}}(\tau_{g_{n}}(f)))...)).$$\\

 Let us calculate $\hat\ce_{n+m_{0}}(\tau_{g_{n}}(f))$. Indeed, from \eqref{E-n1} one gets
\begin{eqnarray*}\hat\ce_{n+m_{0}}(\tau_{g_{n}}(f)\otimes h_{W_{n+m_{0}+1}})&=& \tr_{n+m_{0}]}(K_{[n+m_{0},n+m_{0}+1]}\tau_{g_{n}}(f)\otimes h_{W_{n+m_{0}+1}}K_{[n+m_{0},n+m_{0}+1]}^*)\\
 &=& \bigotimes\limits_{x\in\L_{n+m_{0}-1}}f_{1}^{(x)}\bigotimes\limits_{x\in W_{n+m_{0}}} \tr_{x]}(  A_{x \vee S(x)}
  f_{1}^{(x)}\otimes h^{(S(x))} A^{*}_{x \vee S(x)}) \\
  &=& \bigotimes\limits_{x\in\L_{n+m_{0}-1}}f_{1}^{(x)}\bigotimes\limits_{x\in W_{m_{0}}} \tr_{g_{n} \circ x]}(  A_{g_{n} \circ x \vee S(g_{n} \circ x)}
  f^{(g_{n} \circ x)}\otimes h^{(S(g_{n} \circ x))} A^{*}_{g_{n} \circ x \vee S(g_{n} \circ x)}) \\
  &&\otimes \bigotimes\limits_{x\in W_{n+m_{0}}\setminus\{g_{n},W_{m_{0}}\}} \tr_{x]}(  A_{x \vee S(x)} h^{(x)}
   A^{*}_{x \vee S(x)})  \\
  &=& \bigotimes\limits_{x\in\L_{n+m_{0}-1}}f_{1}^{(x)}\bigotimes\limits_{x\in W_{m_{0}}} T_{m_{0}}^{(g_{n}\circ x)}\otimes\bigotimes\limits_{x\in W_{n+m_{0}}\setminus\{g_{n},W_{m_{0}}\}} h^{(x)}
 \end{eqnarray*}
where $$T_{m_{0}}^{(g_{n}\circ x)}=\tr_{g_{n} \circ x]}(  A_{(g_{n} \circ x) \vee S(g_{n} \circ x)}
  f^{(g_{n} \circ x)}\otimes h^{(S(g_{n} \circ x))} A^{*}_{(g_{n} \circ x )\vee S(g_{n} \circ x)}).$$

Hence,

\begin{eqnarray*}
\ce_{n+m_{0}-1}(\id\otimes \hat\ce_{n+m_{0}}(\tau_{g_{n}}(f)))&=& \bigotimes\limits_{x\in\L_{n+m_{0}-2}}f_{1}^{(x)} \tr_{n+m_{0}-1]}(K_{[n+m_{0}-1,n+m_{0}]} \bigotimes\limits_{x\in W_{n+m_{0}-1}}f_{1}^{(x)}\bigotimes\limits_{x\in W_{m_{0}}} T_{m_{0}}^{(g_{n}\circ x)}\\[2mm]&&\otimes\bigotimes\limits_{x\in W_{n+m_{0}}\setminus\{g_{n},W_{m_{0}}\}} h^{(x)}K_{[n+m_{0}-1,n+m_{0}]}^*)\\[2mm]
&=&  \bigotimes\limits_{x\in\L_{n+m_{0}-2}}f_{1}^{(x)} \bigotimes\limits_{x\in W_{m_{0}-1}}\tr_{g_{n} \circ x]}\bigg(  A_{(g_{n} \circ x )\vee S(g_{n} \circ x)} f^{(g_{n} \circ x)}\\[2mm]
&&\otimes T_{m_{0}}^{(S(g_{n} \circ x))}\otimes A^{*}_{(g_{n} \circ x) \vee S(g_{n} \circ x)}\bigg)\\[2mm]
&&\otimes\bigotimes\limits_{x\in W_{n+m_{0}-1}\setminus\{g_{n},W_{m_{0}-1}\}} \tr_{x]}(  A_{x \vee S(x)}h^{(x)} A^{*}_{x \vee S(x)}) \\[2mm]
&=&\bigotimes\limits_{x\in\L_{n+m_{0}-2}}f_{1}^{(x)}\bigotimes\limits_{x\in W_{m_{0}-1}} T_{m_{0}-1}^{(g_{n}\circ x)}\otimes\bigotimes\limits_{x\in W_{n+m_{0}-1}\setminus\{g_{n},W_{m_{0}}\}} h^{(x)}
\end{eqnarray*}
where $$T_{m_{0}-1}^{(g_{n}\circ x)}=\tr_{g_{n} \circ x]}\left(  A_{g_{n} \circ x \vee S(g_{n} \circ x)} f^{(g_{n} \circ x)}\otimes T_{m_{0}}^{(S(g_{n} \circ x))}\otimes A^{*}_{g_{n} \circ x \vee S(g_{n} \circ x)}\right).$$

By iteration, we obtain
\begin{eqnarray*}\ce_{n}(\id \otimes\ce_{n+1} (\id \otimes...\otimes \ce_{n+m_{0}-1}( \id \otimes\hat\ce_{n+m_{0}}(\tau_{g_{n}}(f)))...)) &=&\bigotimes\limits_{x\in W_{0}} T_{0}^{(g_{n}\circ x)}\otimes\bigotimes\limits_{x\in W_{n}\setminus\{g_{n},W_{0}\}} h^{(x)}\\
 &=&T_{0}^{(g_{n}\circ x_{0})}\otimes\bigotimes\limits_{x\in W_{n}\setminus\{g_{n}\circ x_{0}\}} h^{(x)},
 \end{eqnarray*}

which yields
\begin{eqnarray*}
  \ffi_{w_{0},\bh}(a_{N_{0}} \otimes \id \otimes \id ...\otimes \id \otimes \tau_{g_{n}}(f))&=&\tr(\omega_{0}\ce_{0}\circ ...\circ \ce_{N_{0}}( a_{N_{0}}\otimes \ce_{N_{0}+1}( \id \\[2mm]
  &&\otimes \ce_{N_{0}+2}( \id \otimes...\otimes \ce_{n-1}(T_{0}^{(g_{n}\circ x_{0})}\otimes\bigotimes\limits_{x\in W_{n}\setminus\{g_{n}\circ x_{0}\}} h^{(x)}))).
  \end{eqnarray*}

Then Lemma \ref{7.3} implies

\begin{eqnarray*}\lim_{n \rightarrow\infty}\ffi_{w_{0},\bh}(a_{N_{0}} \otimes \id \otimes \id ...\otimes \id \otimes \tau_{g_{n}}(f)) &=&\ffi_{w_{0},\bh}(a_{N_{0}})\tr(\omega_{0}\ce_{0}\circ ... \circ \ce_{n-1}(T_{0}^{(g_{n}\circ x_{0})}\otimes\bigotimes\limits_{x\in W_{n}\setminus\{g_{n}\circ x_{0}\}} h^{(x)})))\\
&=&\ffi_{w_{0},\bh}(a_{N_{0}})\ffi_{w_{0},\bh}(\tau_{g_{n}}(f)).
\end{eqnarray*}

This completes the proof.
\end{proof}

\subsection{Non quasi equivalence of $\varphi_1$ and
$\varphi_2$}

In this subsection we are going to prove that the states
$\varphi_1$ and $\varphi_2$ are not quasi equivalent.
 To establish the
non-quasi equivalence, we are going to use the following result
(see \cite[Corollary 2.6.11]{BR}).
\begin{theorem}\label{br-q}
Let $\varphi_1,$ $\varphi_2$ be two factor states on a quasi-local
algebra $\ga=\cup_{\Lambda}\ga_\Lambda$. The states $\varphi_1,$
$\varphi_2$ are  quasi-equivalent if and only if for any given
$\varepsilon>0$ there exists a finite volume $\Lambda\subset L$
such that $\|\varphi_1(a)-\varphi_2(a)\|<\varepsilon \|a\|$ for
all $a\in B_{\Lambda^{'}}$ with
$\Lambda^{'}\cap\Lambda=\emptyset.$
\end{theorem}

Now due to Theorem \ref{CPr} the states $\varphi_1$ and $\varphi_2$ have clustering property, and hence they are factor states.
Let us define an element of $\mathcal{B}_{\Lambda_n}$ as follows:
$$
E_{\Lambda_n}:=e_{11}^{x_{W_n}^{(1)}}\otimes\bigg(\bigotimes_{y\in
\Lambda_n\setminus \{x_{W_n}^{(1)}\}}\id^{y}\bigg),
$$
where $x_{W_n}^{(1)}$ is defined in \eqref{xw}.
Now we are going to calculate $\varphi_1(E_{\Lambda_n})$ and
$\varphi_2(E_{\Lambda_n})$, respectively.
First consider the state $\varphi_1$, then we know that this state
is defined by  $\omega_0=\frac{1}{\xi_0}\id$ and
$h^{x}=h=\xi_0\id+\xi_3\sigma_{z}$. Define two  elements of
$\mathcal{B}_{W_n}$ by
$$\hat{\bh}_n:=\id^{x_{W_n}^{(1)}}\otimes\bigotimes_{x\in W_n\setminus \{x_{W_n}^{(1)}\}}h^{(x)}$$
$$\check{\bh}_{n}:=\sigma_{z}^{x_{W_n}^{(1)}}\otimes\bigotimes_{x\in W_n\setminus \{x_{W_n}^{(1)}\}}h^{(x)}$$
\begin{lemma}\label{f-p-11}
Let
$$\hat{\psi}_n:=\Tr_{n-1]}\big[\omega_{0}K_{[0,1]}\cdots K_{[n-1,n]}\hat{\bh}_{n}K_{[n-1,n]}^{*}\cdots K_{[0,1]}^{*}\big]$$
$$\check{\psi}_n:=\Tr_{n-1]}\big[\omega_{0}K_{[0,1]}\cdots K_{[n-1,n]}\check{\bh}_{n}K_{[n-1,n]}^{*}\cdots K_{[0,1]}^{*}\big]$$
Then there are two pairs of reals
$(\hat{\rho}_{1},\hat{\rho}_{2})$ and
$(\check{\rho}_{1},\check{\rho}_{2})$ depending on $\theta$ such
that
$$\left\{
  \begin{array}{ll}
    \hat{\psi}_n=\hat{\rho}_{1}+\hat{\rho}_{2}(\frac{C_1}{C_3}-1)^{n}, \\
\\
    \check{\psi}_n=\check{\rho}_{1}+\check{\rho}_{2}(\frac{C_1}{C_3}-1)^{n} \\
  \end{array}
\right.$$
\end{lemma}
\begin{proof} One can see that
\begin{eqnarray*}
\left(
  \begin{array}{ll}
    \hat{\psi}_n \\
\\
    \check{\psi}_n
  \end{array}
\right) &=& \left(
     \begin{array}{c}
       \Tr_{n-1]}\big[\omega_{0}K_{[0,1]}\cdots K_{[n-2,n-1]}\Tr_{n-1]}[K_{[n-1,n]}\hat{\bh}_{n}K_{[n-1,n]}^{*}]K_{[n-2,n-1]}^{*}\cdots
       K_{[0,1]}^{*}\big],  \\
       \\
       \Tr_{n-1]}\big[\omega_{0}K_{[0,1]}\cdots K_{[n-2,n-1]}Tr_{n-1]}[K_{[n-1,n]}\check{\bh}_{n}K_{[n-1,n]}^{*}]K_{[n-2,n-1]}^{*}\cdots
       K_{[0,1]}^{*}\big]
     \end{array}
\right).
\end{eqnarray*}
After small calculations, we find
 $$\left\{
  \begin{array}{ll}
  \Tr_{x]}\left[A_{(x,(x,1),(x,2))}\big(\id^{(x)}\otimes\id^{(x,1)}\otimes
  h^{(x,2)}\big)A_{(x,(x,1),(x,2))}\right]=C_{1}\xi_0\id^{(x)}+\frac{1}{2}C_{3}\xi_3\sigma_{z}^{(x)} \\
   \\
   \Tr_{x]}\big[A_{(x,(x,1),(x,2))}\big(\id^{(x)}\otimes\sigma^{(x,1)}\otimes
   h_{(\xi_{0},\xi_{3})}^{(x,2)}\big)A_{(x,(x,1),(x,2))}\big]=C_2\xi_3\id^{(x)}+\frac{1}{2}\sigma_{z}^{(x)}
  \end{array}
\right.$$
Hence,  one gets
$$\left\{
\begin{array}{ll}
\Tr_{n-1]}K_{[n-1,n]}\hat{\bh}_{n}K_{[n-1,n]}^{*}=C_1\xi_0\hat{h}_{n-1}+\frac{1}{2}C_3\xi_3\check{h}_{n-1},\\
\\
\Tr_{n-1]}K_{[n-1,n]}\check{\bh}_{n}K_{[n-1,n]}^{*}=C_2\xi_3\hat{h}_{n-1}+\frac{1}{2}\check{h}_{n-1}.
\end{array}
\right.$$
Therefore,
\begin{eqnarray*}
\left(
  \begin{array}{ll}
    \hat{\psi}_n \\[2mm]
    \check{\psi}_n
  \end{array}
\right)
&=&
\left(
\begin{array}{c}
C_1\xi_0\hat{\psi}_{n-1}+\frac{1}{2}C_3\xi_3\check{\psi}_{n-1}
\\[2mm]
C_2\xi_3\hat{\psi}_{n-1}+\frac{1}{2}\check{\psi}_{n-1}
\end{array}
\right)\nonumber\\[2mm]
&=&
\left(
   \begin{array}{cc}
   C_1\xi_0 & \        \ \frac{1}{2}C_3\xi_3 \\[2mm]
     C_2\xi_3 &\        \ \frac{1}{2} \\
     \end{array}
     \right)
\left(
\begin{array}{c}
  \hat{\psi}_{n-1} \\[2mm]
  \check{\psi}_{n-1} \\
\end{array}
\right)\nonumber\\
\vdots\nonumber\\
&=&\left(
   \begin{array}{cc}
   C_1\xi_0 & \        \ \frac{1}{2}C_3\xi_3 \\[2mm]
     C_2\xi_3 &\        \ \frac{1}{2} \\
     \end{array}
     \right)^{n}
\left(
\begin{array}{c}
  \hat{\psi}_{0} \\[2mm]
  \check{\psi}_{0} \\
\end{array}
\right),
\end{eqnarray*}
where $$\left\{
\begin{array}{c}
  \hat{\psi}_{0}=\Tr (\omega_{0})=\frac{1}{\xi_0}\\[2mm]
  \check{\psi}_{0}=\Tr (\omega_{0}\sigma_z) = 0 \\
\end{array}
\right.
$$
The matrix
$$N := \left(
   \begin{array}{cc}
   C_1\xi_0 & \        \ \frac{1}{2}C_3\xi_3 \\[2mm]
     C_2\xi_3 &\        \ \frac{1}{2} \\
     \end{array}
     \right)
     $$
     can be written in diagonal form by:
     $$N=P
     \left(
   \begin{array}{cc}
   1  & \        \ 0 \\
     0 &\        \ \frac{C_{1}}{C_3}-\frac{1}{2} \\
     \end{array}
     \right)
     P^{-1}
     $$
     where
     $$
       P=\left(
   \begin{array}{cc}
   \frac{C_3}{2 C_2} & \        \  -\frac{\xi_3}{\xi_0} \\
     \frac{\xi_3}{\xi_0} &\        \ 1 \\
     \end{array}
     \right), \     \  \det(P)=\frac{3 C_{3}-2 C_{1}}{2 C_2}$$
 So,
     \begin{eqnarray*}
     \left(
       \begin{array}{c}
         \hat{\psi}_n \\
               \check{\psi}_n \\
       \end{array}
     \right)
     &=&P \left(
   \begin{array}{cc}
   1  & \        \ 0 \\
     0 &\        \ (\frac{C_{1}}{C_3}-\frac{1}{2} )^{n}\\
     \end{array}
     \right)P^{-1}
     \left(
     \begin{array}{c}
       \frac{1}{\xi_0} \\
       \\
       0
      \end{array}
     \right)\nonumber\\
      &=&
     \left(
  \begin{array}{ll}
    \hat{\rho}_{1}+\hat{\rho}_{2}(\frac{C_1}{C_3}-\frac{1}{2})^{n}\\[2mm]
    \check{\rho}_{1}+\check{\rho}_{2}(\frac{C_1}{C_3}-\frac{1}{2})^{n} \\
  \end{array}
\right).
\end{eqnarray*}
where
\begin{eqnarray}\label{r1}
&& \hat{\rho}_{1}=\frac{C_{3}^{2}}{3C_{3}-2C_{1}},  \ \
\hat{\rho}_{2}=\frac{2 C_{3}(C_3-C_1)}{3C_{3}-2C_{1}},\\[2mm]
&& \check{\rho}_{1}=\frac{2C_2 C_{3}^{2}\xi_{3}}{
3C_3-2C_1},\      \   \check{\rho}_{2}=-\frac{2C_2
C_{3}^{2}\xi_{3}}{ 3C_3-2C_1}.\label{r2}
\end{eqnarray}
This completes the proof.
\end{proof}
\begin{proposition}\label{6.6} For each $n\in\bn$ one has
\begin{eqnarray*}
\varphi_1(E_{\Lambda_n})&=&\frac{1}{2}\left[(\xi_0+\xi_3)(C_1\xi_0 +C_2\xi_3
)\hat{\rho}_{1}+\frac{C_3}{2}(\xi_0+\xi_3)^{2}\check{\rho}_{1}\right]\nonumber\\
&&+\frac{1}{2}\left[(\xi_0+\xi_3)(C_1\xi_0 +C_2\xi_3
)\hat{\rho}_{2}+\frac{1}{2}C_3(\xi_0+\xi_3)^{2}\check{\rho}_{2}\right]\bigg(\frac{C_1}{C_3}-\frac{1}{2}\bigg)^{n-1}
\end{eqnarray*}
\end{proposition}
\begin{proof} From \eqref{F1} we have
\begin{eqnarray*}
\varphi_1(E_{\Lambda_n})&=&\Tr \big(\omega_0 K_{[0,1]}\cdots K_{[n-1,n]}\bh_{n}^{1/2} E_{\Lambda_n} \bh_{n}^{1/2} K_{[n-1,n]}^{*}\cdots K_{[0,1]}^{*}\big)\\
&=&\Tr \big(\omega_0 K_{[0,1]}\cdots K_{[n-1,n]}\bh_{n} E_{\Lambda_n}  K_{[n-1,n]}^{*}\cdots K_{[0,1]}^{*}\big)
\end{eqnarray*}
One can calculate that
\begin{eqnarray*}
\Tr_{n-1]} (K_{[n-1,n]}\bh_n
K_{[n-1,n]}^{*}E_{\Lambda_n})&=&\Tr_{x_{W_{n-1}}^{(1)}]}\bigg(A_{(x_{W_{n-1}}^{(1)},
x_{W_n}^{(1)}, x_{W_n}^{(2)})}
\big(\id^{x_{W_{n-1}}^{(1)}}\otimes e_{1,1}h^{x_{W_n}^{(1)}}\otimes h^{x_{W_n}^{(2)}}\big)\nonumber\\
&&A_{(x_{W_{n-1}}^{(1)}, x_{W_n}^{(1)}, x_{W_n}^{(2)})}^{*}\bigg)\otimes\bigotimes_{x\in W_{n-1}\setminus \{x_{W_n}^{(1)}\}}h^{(x)}\nonumber\\
&=&\frac{1}{2}\left[(\xi_0+\xi_3)(C_1\xi_0+C_2\xi_3)\hat{\bh}_{n-1}+\frac{C_3}{2}(\xi_0+\xi_3)^{2}\check{\bh}_{n-1}\right].
\end{eqnarray*}

Hence
\begin{eqnarray*}
\varphi_1(E_{\Lambda_n})
&=&\frac{1}{2}(\xi_0+\xi_3)(C_1\xi_0+ C_2\xi_3)\Tr\left[ \omega_0K_{[0,1]}\cdots K_{[n-2,n-1]}\hat{\bh}_{n-1} K_{[n-2,n-1]}^{*}\cdots
K_{[0,1]}^{*}\right]\nonumber\\
&&+C_3\bigg(\frac{\xi_0+\xi_3}{2}\bigg)^{2}\Tr\left[
\omega_0 K_{[0,1]}\cdots K_{[n-2,n-1]}\check{\bh}_{n-1}
K_{[n-2,n-1]}^{*}\cdots
 K_{[0,1]}^{*}\right].\nonumber\\
&=&\frac{1}{2}\left[(\xi_0+\xi_3)(C_1\xi_0+C_2\xi_3)\hat{\psi}_{n-1}+\frac{C_3}{2}(\xi_0+\xi_3)^{2}\check{\psi}_{n-1}\right].
\end{eqnarray*}
Now using the values of $\hat{\psi}_{n-1}$  and
$\check{\psi}_{n-1}$ given by the previous lemma we obtain the
result.
\end{proof}
 Now we consider the state $\varphi_2$. Recall that this state is defined by  $\omega_0=\frac{1}{\xi_0}\id$ and
$h^{x}=h'=\xi_0\id-\xi_3\sigma_z$. Define two elements of
$\mathcal{B}_{W_n}$ by
$$\hat{\bh'}_n:=\id^{x_{W_n}^{(1)}}\otimes\bigotimes_{x\in W_n\setminus \{x_{W_n}^{(1)}\}}h'^{(x)}$$
$$\check{\bh'}_{n}:=\sigma^{x_{W_n}^{(1)}}\otimes\bigotimes_{x\in W_n\setminus \{x_{W_n}^{(1)}\}}h'^{(x)}$$
Using the same argument like in the proof of Lemma \ref{f-p-11} we
can prove the following auxiliary fact.
\begin{lemma}\label{6.7} Let
$$\hat{\phi}_n:=\Tr_{n-1]}\big[\omega_{0} K_{[0,1]}\cdots K_{[n-1,n]}\hat{\bh'}_{n}K_{[n-1,n]}^{*}\cdots K_{[0,1]}^{*}\big]$$
$$\check{\phi}_n:=\Tr_{n-1]}\big[\omega_{0} K_{[0,1]}\cdots K_{[n-1,n]}\check{\bh'}_{n}K_{[n-1,n]}^{*}\cdots K_{[0,1]}^{*}\big]$$
Then there are two pairs of reals  $(\hat{\pi}_{1},\hat{\pi}_{2})$
and $(\check{\pi}_{1},\check{\pi}_{2})$ depending on $\theta$ such
that
$$\left\{
  \begin{array}{ll}
    \hat{\phi}_n=\hat{\pi}_{1}+\hat{\pi}_{2}(\frac{C_1}{C_3}-\frac{1}{2})^{n}, \\
    \check{\phi}_n=\check{\pi}_{1}+\check{\pi}_{2}(\frac{C_1}{C_3}-\frac{1}{2})^{n} \\
  \end{array}
\right.$$ where
$$\hat{\pi}_{1}=\frac{C_3^2}{3 C_3-2 C_1},  \      \   \hat{\pi}_{2}=\frac{2 C_{3}(C_3-C_1)}{3C_{3}-2C_1},  $$
$$ \check{\pi}_{1}=-\frac{2C_{2}C_3^2\xi_{3}}{ 3C_3-2C_1},\      \   \check{\pi}_{2}=\frac{2C_{2}C_3^{2}\xi_{3}}{ 3C_3-2C_1}.$$
\end{lemma}
\begin{proposition}\label{6.8} For each $n\in\bn$ one has
\begin{eqnarray*}
\varphi_2(E_{\Lambda_n})&=&\frac{1}{2}\left[(\xi_0-\xi_3)(C_1\xi_0-C_2\xi_3)\hat{\pi}_{1}+\frac{C_3}{2}(\xi_0-\xi_3)^{2}\check{\pi}_{1}\right]\\[2mm]
&&+\frac{1}{2}\left[(\xi_0-\xi_3)(C_1\xi_0 -C_2\xi_3)\hat{\pi}_{2}+\frac{C_3}{2}(\xi_0-\xi_3)^{2}\check{\pi}_{2}\right]
\bigg(\frac{C_1}{C_3}-\frac{1}{2}\bigg)^{n-1}.
\end{eqnarray*}
\end{proposition}
\begin{proof} From \eqref{F2} we find
\begin{eqnarray}\label{6.81}
\varphi_2(E_{\Lambda_n})&=&\Tr \big(\omega_0 K_{[0,1]}\cdots K_{[n-1,n]}\bh_{n}^{'1/2} E_{\Lambda_n} \bh_{n}^{'1/2} K_{[n-1,n]}^{*}\cdots K_{[0,1]}^{*}\big)\\
&=&\Tr \big(\omega_0 K_{[0,1]}\cdots K_{[n-1,n]}\bh_{n}^{'} E_{\Lambda_n} K_{[n-1,n]}^{*}\cdots K_{[0,1]}^{*}\big)\\
\end{eqnarray}
We easily calculate that
\begin{eqnarray*}
\Tr_{n-1]} (K_{[n-1,n]}\bh'_n E_{\Lambda_n}
K_{[n-1,n]}^{*})=\frac{1}{2}(\xi_0-\xi_3)(C_1\xi_0-C_2\xi_3)\hat{\bh'}_{n-1}+C_3\bigg(\frac{\xi_0-\xi_3}{2}\bigg)^{2}\check{\bh'}_{n-1}.
\end{eqnarray*}
Hence, from  \eqref{6.81} one gets
\begin{eqnarray*}
\varphi_2(E_{\Lambda_n})
=\frac{1}{2}\left[(\xi_0-\xi_3)(C_1\xi_0-C_2\xi_3)\hat{\phi}_{n-1}+\frac{C_3}{2}(\xi_0-\xi_3)^{2}\check{\phi}_{n-1}\right].
\end{eqnarray*}
Using the values of $\hat{\phi}_{n-1}$  and $\check{\phi}_{n-1}$
given in Lemma \ref{6.7}, we obtain the desired assertion.
\end{proof}
\begin{theorem}\label{6.9}
Assume that $J \in ]-J_0,J_0[$, then the two Backward QMC $\varphi_1$ and $\varphi_2$ are not quasi-equivalent.
\end{theorem}
\begin{proof}
For any  $\forall n\in\natural$ it is clear that
$E_{\Lambda_n}\in\mathcal{B}_{\Lambda_n}\setminus\mathcal{B}_{\Lambda_{n-1}}.$
Therefore, for any finite subset $\Lambda\in L$, there exists
$n_0\in\natural$ such that $\Lambda\subset\Lambda_{n_0}$.  Then
for all $n>n_0$ one has
$E_{\Lambda_n}\in\mathcal{B}_{\Lambda_n}\setminus\mathcal{B}_{\Lambda}.$
It is clear that
$$\|E_{\Lambda_n}\|=\|e_{1,1}^{x_{W_n}^{(1)}}\bigotimes_{y\in
L\setminus \{x_{W_n}^{(1)}\}}\id^{y}\|=\|e_{1,1}\|=\frac{1}{2}.$$
From Propositions \ref{6.6} and \ref{6.8} we obtain
\begin{eqnarray*}
\left|\varphi_1(E_{\Lambda_n})-\varphi_1(E_{\Lambda_n})\right|
&=&\frac{1}{2}\bigg|\big[(\xi_0+\xi_3)(C_1\xi_0 +C_2\xi_3)\hat{\rho}_{1}
+C_3(\xi_0+\xi_3)^{2}\check{\rho}_{1}\big]\\[2mm]
&&-\big[(\xi_0-\xi_3)(C_1\xi_0 -C_2\xi_3 )\hat{\pi}_{1}
+C_3(\xi_0-\xi_3)^{2}\check{\pi}_{1}\big]\nonumber\\
&&+\bigg(\big[(\xi_0+\xi_3)(C_1\xi_0 +C_2\xi_3)\hat{\rho}_{2}+C_3(\xi_0+\xi_3)^{2}\check{\rho}_{2}\big]\\[2mm]
&&-\big[(\xi_0-\xi_3)(C_1\xi_0 - C_2\xi_3 )\hat{\pi}_{2}+C_3(\xi_0-\xi_3)^{2}\check{\pi}_{2}\big]\bigg)
\left(\frac{C_1}{C_3}-\frac{1}{2}\right)^{n-1}\bigg|\nonumber\\
&\geq& I_1-I_2\left|\frac{C_1}{C_3}-\frac{1}{2}\right|^{n-1}
\end{eqnarray*}
where
\begin{eqnarray*}
I_1&=&\frac{1}{2}\bigg|\big[(\xi_0+\xi_3)(C_1\xi_0 +C_2\xi_3 )\hat{\rho}_{1}+C_3(\xi_0+\xi_3)^{2}\check{\rho}_{1}\big]\\[2mm]
&&-\big[(\xi_0-\xi_3)(C_1\xi_0 - C_2\xi_3 )\hat{\pi}_{1}+C_3(\xi_0-\xi_3)^{2}\check{\pi}_{1}\big]\bigg|\\[2mm]
I_2&=&\frac{1}{2}\bigg|\big[(\xi_0+\xi_3)(C_1\xi_0 +C_2\xi_3 )\hat{\rho}_{2}+C_3(\xi_0+\xi_3)^{2}\check{\rho}_{2}\big]\\[2mm]
&&-\big[(\xi_0-\xi_3)(C_1\xi_0 - C_2\xi_3 )\hat{\pi}_{2}+C_3(\xi_0-\xi_3)^{2}\check{\pi}_{2}\big]\bigg|.
\end{eqnarray*}
Due to $\beta> 0, \theta=\exp2\beta>1$,  $C_1>0$,$C_3>0,
\xi_0>, \xi_3>0$, one can find that
 \begin{eqnarray*}
I_1=\frac{C_3\xi_3(2 C_2+C_3)}{3 C_3-2 C_1}>0.
 \end{eqnarray*}
Now we have  $\frac{2C_1-C_3}{2C_3} =\frac{\theta^{J_0}(\theta^J+\theta^{-J})+2}{2(\theta^{2J_0}-1)}$, since $J \in ]-J_0,J_0[$ then   $$\frac{2C_1-C_3}{2C_3}
=\frac{\theta^{J_0}(\theta^J+\theta^{-J})+2}{2(\theta^{2J_0}-1)}\sim \frac{1}{2\theta^{(J_0-J)}}\leq \frac{1}{2},~~~~\theta \geq \theta_0$$
Then the following equality is hold:
\begin{eqnarray*}
\bigg|\frac{C_1}{C_3}-\frac{1}{2}\bigg|\leq \frac{1}{2}
\end{eqnarray*}
which yields
\begin{equation*}
I_2\left|\frac{C_1}{C_3}-\frac{1}{2}\right|^{n-1}\rightarrow
0 \ \ \textrm{as} \ \ n\rightarrow +\infty.
\end{equation*}
Then there exists $n_1\in \bn$ such that $\forall n\geq n_0$ one
has
\begin{equation*}
I_2\left|\frac{C_1}{C_3}-\frac{1}{2}\right|^{n}\leq
\frac{\varepsilon_1 }{2}.
\end{equation*}
Hence, for all $n\geq n_1$ we obtain
\begin{equation*}
\left|\varphi_1(E_{\Lambda_n})-\varphi_1(E_{\Lambda_n})\right|\geq
\frac{\varepsilon_1}{2}=\varepsilon_1\|E_{\Lambda_n}\|.
\end{equation*}
 This, according to Theorem \ref{br-q}, means that the states
 $\varphi_1$ and $\varphi_2$ are not quasi-equivalent.
The proof is complete.
\end{proof}
Now Theorems \ref{5.3},\ref{6.2} and \ref{6.9} imply Theorem
\ref{Main}.


\section{ QMC associated with the XY-interaction model with $J_0=0$ }

In this section, we consider a model which does not contain the classical Ising part, i.e. $J_0=0$, which means the model has only competing
XY-interactions. In this setting, from \eqref{K1} one gets

\begin{eqnarray*}\label{Axn}
A_{(u,(u,1),(u,2))}&=&L_{>(u,1),(u,2)<}\\
&&=\id^{(u,1)}\otimes \id^{(u,2)} + \sinh(J\beta) H_{>(u,1),(u,2)<}\\
&&+(\cosh(J\beta)-1)H^{2}_{>(u,1),(u,2)<}\\
&&=R_{1}\id^{(u)}\otimes\id^{(u,1)}\otimes\id^{(u,2)}
+R_2\id^{(u)}\otimes\sigma_{x}^{(u,1)}\otimes\sigma_{x}^{(u,2)}\\
&&+R_{2}\id^{(u)}\otimes\sigma_{y}^{(u,1)}\otimes\sigma_{y}^{(u,2)}
+R_3\id^{(u)}\otimes\sigma_{z}^{(u,1)}\otimes \sigma_{z}^{(u,2)}
\end{eqnarray*}
where
$$\left\{
  \begin{array}{ll}
    R_1=\frac{1}{4}(\cosh(J\beta)+1);\bigskip \\
    R_2=\frac{\sinh(J\beta)}{2};\bigskip  \\
    R_3=\frac{1}{2}(1-\cosh(J\beta)).
  \end{array}
\right.$$

Therefore, one finds:

\begin{eqnarray} \label{eqdern}
h&=&Tr_{x]}A_{(u,(u,1),(u,2))}[\id^{(u)}\otimes h\otimes h]A_{(u,(u,1),(u,2))}^{*}\nonumber\\
&=&[(R_1+2R_1^2+R_3^2) \tr(h)^2]\id^{(u)}.
\end{eqnarray}

The equation \eqref{eqdern}
is reduced to the following one
\begin{equation}\label{EQ1}
\left\{
   \begin{array}{lll}
h_{11}=h_{22}=\frac{1}{(R_1+2R_1^2+R_3^2)},\\
h_{21}=0, h_{12}=0.\\
   \end{array}
 \right.
 \end{equation}

Then putting $\alpha=\frac{1}{(R_1+2R_1^2+R_3^2)}$ we get
\begin{equation}
h_{\alpha}=
 \left(
  \begin{array}{cc}
    \alpha & 0 \\
    0 & \alpha \\
  \end{array}
\right)
\end{equation}

\begin{proposition}\label{5.1}
The pair $(\omega_{0},\{h^{x}=h_{\alpha}| x\in L\})$ with
$\omega_{0}=\frac{1}{\alpha}\id,\   \  h^{x}=h_{\alpha}, \forall
x\in L,$ is solution of \eqref{eq1},\eqref{eq2}. Moreover the
associated Backward QMC  can be written on the local algebra
$\mathcal{B}_{L, loc}$ by:
\begin{equation}
\varphi_{\alpha}(a)=\alpha^{2^{n}-1}\Tr\bigg(\prod_{i=0}^{n-1}K_{[i,i+1]} a \prod_{i=0}^{n-1} K_{[n-i-1,n-i]}^{*}\bigg),
\ \ \forall a\in B_{\Lambda_{n}}.
\end{equation}
In this case, there is no phase transition.
\end{proposition}

We stress that if one takes nearest neighbor XY interactions on the Cayley tree of order two, still there does not occur a phase transition \cite{AMSa1}. However, if the order of the tree is three or more then for the mentioned model there exists a phase transition
\cite{AMSa2}.

\section*{Acknowledgments}
The authors are grateful to professors L.Accardi  for
fruitful discussions and useful suggestions on the definition of
the phase transition.

\end{document}